\documentclass[12pt,a4paper,reqno]{amsart}
\usepackage{amsthm,amssymb,amsmath,latexsym}
\usepackage{mathrsfs,setspace,pstricks}
\usepackage{epsfig, subcaption, graphics,graphicx}

\numberwithin{equation}{section}
\input colordvi

\newcommand\new[1]{}

\advance\textwidth by 2.5cm 
\advance\oddsidemargin by -1.6cm
\advance\evensidemargin by -1.6cm

\newtheorem{Theorem}{Theorem}[section]
\newtheorem{Definition}[Theorem]{Definition}
\newtheorem{Proposition}[Theorem]{Proposition}
\newtheorem{Lemma}[Theorem]{Lemma}

\newtheorem{Remark}[Theorem]{Remark}

\newcommand{\no}{\nonumber}
\newcommand{\noi}{\noindent}

\newcommand{\pa}{\partial}
\newcommand{\ds}{\displaystyle}
\newcommand{\RR}{\mathbb{R}}
\newcommand{\si}{\sigma}

\newcommand{\la}{\lambda}

\newcommand{\vf}{\varphi}

\newcommand{\vp}{\varepsilon}

\newcommand{\Int}{\displaystyle\int\limits}

\allowdisplaybreaks

\date{}

\begin{document}
	
	\pagenumbering{arabic}
	
	\author{Milan Kumar Das}
	\address{IISER Pune, India}
	\email{milankumar.das@students.iiserpune.ac.in}
	\thanks{The first author acknowledges SRF grant of UGC}
	
	\author{Anindya Goswami}
	\address{IISER Pune, India}
	\email{anindya@iiserpune.ac.in}\thanks{The second author acknowledges SERB Grant EMR/2016/004810}
	
	\author{Nimit Rana}
	\address{University of York, UK}
	\email{nr734@york.ac.uk}
	
	\title[Risk Sensitive Optimization in a Switching Jump Diffusion Market]{Risk Sensitive Portfolio Optimization in a Jump Diffusion Model with Regimes}
	
	\addtocounter{footnote}{-1} \vskip 1 true cm
	
	\begin{abstract}
		This article studies a portfolio optimization problem, where the market consisting of several stocks is modeled by a multi-dimensional jump diffusion process with age-dependent semi-Markov modulated coefficients. We study risk sensitive portfolio optimization on the finite time horizon. We study the problem by using a probabilistic approach to establish the existence and uniqueness of the classical solution to the corresponding Hamilton-Jacobi-Bellman (HJB) equation. We also implement a numerical scheme to investigate the behavior of solutions for different values of the initial portfolio wealth, the maturity, and the risk of aversion parameter. \end{abstract}

	\maketitle
	
	
	{\bf AMSC: } 91G10, 93E20, 60K15, 60H10
	

	\section{Introduction} \label{sec:Intro}
	\noi Following the seminal work of Markowitz \cite{MA}, the 	problem of optimization of an investor's portfolio based on 	different criteria and market assumptions are being studied by 	several authors. In the mean-variance optimization approach, as 	done by Markowitz, either the expected value of the portfolio wealth 	is optimized by keeping the variance fixed, or the variance is 	minimized by keeping the expectation fixed. Though the Markowitz's 	mean-variance approach to the portfolio optimization is immensely useful in 	practice, its scope is limited by the fact that only Gaussian 	distributions are completely determined by their first two 	moments. In a pioneering work, Merton \cite{M1}, \cite{M2} has 	introduced the utility maximization to the optimal portfolio 	selection. Merton's approach is based on applying the method of 	stochastic optimal control via an appropriate Hamilton-Jacobi- 	Bellman (HJB) equation. The corresponding optimal dynamic 	portfolio allocation can also be obtained from the same equation. Although this approach has greater mathematical tractability but does not capture the tradeoff between maximizing expectation and minimizing the variance of the portfolio value. 
	
	There is another approach, namely the risk sensitive optimization where  a tradeoff between the long run expected growth rate and the asymptotic variance is captured in implicitly. The aforesaid utility maximization method can be employed to study the 	risk-sensitive optimization by choosing a parametric family of exponential utility functions. In such optimization, an appropriate value of the parameter is to be chosen by the investor depending on the investors degree of risk tolerance. We refer 	\cite{bp1}, \cite{fleming1}, \cite{fleming2}, \cite{LE} for this 	criterion under the geometric Brownian motion (GBM) market model.
	
	Risk sensitive optimization of portfolio value in a more general type of market is also studied by various authors. The jump diffusion model is one of such generalizations which captures the discontinuity of asset dynamics. The empirical results support such models \cite{DMV}. Terminal utility optimization problem under such a model assumption is studied in \cite{KJ}. In all these references, it is assumed that the market parameters, i.e., the coefficients in the asset price dynamics, are either constant or deterministic functions of time. We study a class of models where these parameters are allowed to be finite state pure jump processes. We call each state of the coefficients as a regime and the dynamics, as a regime switching model. The regime switching can be of various types. It is known that for a Markov switching model, the sojourn or holding times in each state are distributed as exponential random variables, whereas the holding time can be any positive random variable for the semi-Markov case. Thus the class of semi-Markov processes subsumes the class of Markov chains. There are some statistical results in the literature (see \cite{BU}, \cite{JH} and the references therein), which emphasize the advantage of the applicability of semi-Markov switching models over simple homogeneous Markov switching models. It is mainly useful to deal with the impact of a changing environment, which exhibits duration dependence. To understand this, consider a market situation where, if the volatility of a certain stock price remains low for longer than certain duration, then that observation discourages increasingly more traders to trade on that, depending on the length of the duration. In that case, this type of duration dependence mass-trading behavior might cause further low volume trading resulting in lack of volatility boost. In this type of market behavior, the density function of holding time of low volatility regime should exhibit heavier tail than exponential. It is important to note that, a Markov chain either time homogeneous or inhomogeneous, does not exhibit such age-dependent transition, whereas a generic semi-Markov process may exhibit this phenomenon. It motivates us to consider the age-dependent transition of the regimes.
	
Risk sensitive portfolio optimization in a GBM model with Markov regimes is studied in \cite{GGS10} 	whereas \cite{GGS09} studies the same problem in a semi-Markov modulated GBM model. In \cite{GGS09} the market parameters, 	$r$, $\mu^l$ and $\sigma^l$ are driven by a finite-state semi-Markov process $\{X_t\}_{t\ge 0}$, where $\mu^l$ 	and $\sigma^l$ denote the drift and volatility parameters of the $l$-th asset in the portfolio. Strictly speaking, 	the assumption  that all the parameters from different assets are governed by a single semi-Markov process is 	 rather restrictive. Ideally, those could be driven by independent or correlated processes in practice. 	Although two independent Markov processes jointly becomes a Markov process, the same phenomena is not 	true for semi-Markov processes. For this reason, the case of independent regimes are important where regimes are not Markov. 	

In general, a pure jump process need not be a semi-Markov process. In particular, the class of 	age-dependent processes (as in \cite{GhS}) is much wider than the type of age independent semi-Markov processes 	studied in \cite{GGS09}. In a recent paper \cite{GPS}, option pricing is studied in a switching market where the regimes are assumed to be an age-dependent process. An age-dependent process $X=\{X_t\}_{t\ge 0}$ on a finite state space $\mathcal{X}:=\{1,2, \ldots,k\}$ is specified by its instantaneous transition rate $\la$, which is a collection of measurable functions 	$\la_{ij}:[0, \infty)\to (0, \infty)$ where $(i,j) \in \mathcal{X}_2$ and $\mathcal{X}_2:=\{(i,j)|i\neq j \in \mathcal{X}\}$. Indeed, 	embedding $\mathcal{X}$ in $\mathbb{R}$, an age-dependent process $X$ on $\mathcal{X}$ is defined as the strong solution to the following system 	of stochastic integral equations (SIEs)
	\begin{equation}\label{agedep}
	\left.\begin{array}{ll}
	X_t &= X_0 + \Int_{(0,t]}\Int_{\mathbb{R}} h_\la(X_{u-}, Y_{u-},z)\wp(du,dz),\\
	Y_t &= Y_0+ t- \Int_{(0,t]} \Int_{\mathbb{R}} g_\la(X_{u-}, Y_{u-},z) \wp(du,dz),
	\end{array}\right\}
	\end{equation}
	where $\wp (dt,dz)$ is the Poisson random measure with
	intensity $dtdz$, independent of $X_0,Y_0$, and
	\begin{eqnarray*}
		h_\la(i,y,z) := \sum_{j \in \mathcal{X}\setminus\{i\}} (j-i) 1_{\Lambda_{ij}(y)}(z), ~~~  g_\la(i,y,z)
		:= \sum_{j \in \mathcal{X}\setminus\{i\}} y 1_{\Lambda_{ij}(y)}(z),
	\end{eqnarray*}
	and for every $y>0$, $\Lambda_{ij}(y):=\left[\sum_{(i',j')\prec(i,j)}\la_{i'j'}(y),~~\la_{ij}(y)+\sum_{(i',j')\prec(i,j)}\la_{i'j'}(y)\right)$, using a strict total order $\prec$ on $\mathcal{X}_2$. In particular $\prec$ can be taken as lexicographic ordering. The existence of unique strong solution of the SIEs \eqref{agedep} follows from (\cite{IW}, Chap. IV, p.231), since $h_\la$ and $g_\la$ are compactly supported in $z$ variable. We refer to \cite{GhS} for a proof that  $\la$ indeed represents the instantaneous transition rate of $X$.
	
	In this paper, we consider a regime switching jump diffusion model of a financial market, where an observed Euclidean space valued pure jump process drives the regimes of every asset. Further, we assume that every component of that pure jump process is an age-dependent semi-Markov process and the components are independent. We study the finite horizon portfolio optimization via the risk sensitive criterion under the above market assumption. The optimization problem is solved by studying the corresponding HJB equation, where we employ the technique of separation of variables to reduce the HJB equation to a system of linear first order PDEs containing some non-local terms. In the reduced equation, the nature of non-locality is such that the standard theory of integro-pde is not applicable to establish the existence and uniqueness of the solution. In this paper, to show well-posedness of this PDE, a Volterra integral equation(IE) of the second kind is obtained and then the existence of a unique $C^1$ solution is shown. Then it is proved that the solution to the IE is a classical solution to the PDE under study. The uniqueness of the PDE is proved by showing that any classical solution also solves the IE. In the uniqueness part, we use conditioning with respect to the transition times of the underlying process. Besides, we also obtain the optimal portfolio selection as a continuous function of time and underlying switching process. The expression of this function does not involve the functional parameter $\la$. Thus the optimal selection is robust. Our approach of solving the PDE also enables us to develop a robust numerical procedure to compute the optimal portfolio wealth using a quadrature method.
	
	The rest of the paper is organized as follows. In the next section, we give a rigorous description of the model of a financial market dynamics and then derive the wealth process of an investor's portfolio. The problem of optimizing the portfolio wealth under the risk sensitive criterion on the finite time horizon is also stated in Section 2. In Section 3 we have established a characterization of the optimal wealth using the corresponding Hamilton-Jacobi-Bellman equation. An optimal portfolio strategy is also shown to exist in the class of Markov feedback control. Furthermore, an optimal feedback control is produced as a minimizer of a certain functional associated with the HJB equation. We illustrate the theoretical results by performing numerical experiments with an example and obtain some relevant results in Section 4. Section 5 contains some concluding remarks. The proofs of certain important lemmata are given in the Appendix.
	
\section{Model Description}
	\label{sec:PrblmForm}
\subsection{Model parameters}
	\noi Let $\mathcal{X}$ denote a finite subset of $\RR$. Without loss of generality, we choose $\mathcal{X} = \{1,2,\ldots,k\}$ and $\mathcal{X}_2 $ as defined above \eqref{agedep}. Consider for each $l=0,1, \ldots, n$, $\la^l: \mathcal{X}_2\times [0,\infty)\to (0,\infty)$ a
	continuously differentiable function in $y$ with $\lambda^l_{ii}(y)=-\sum_{j\neq i}\lambda^l_{ij}(y)$ and
	\begin{equation}\label{eq:lambda}
	\no \underset{y\rightarrow\infty}{\lim}\Lambda_i^l(y)=\infty ,\text{where}\, \Lambda_i^l(y):=\Int_0^y \sum_{j\neq i} \lambda_{ij}^l(v)dv.
	\end{equation}
	Assume that for each $j=1,2,\ldots, m_2$, $\nu_j$ denotes a finite Borel measure on $\RR$. Let $r:[0,T]\times \mathcal{X}^{n+1} \rightarrow [0,\infty)$, $\mu^l:[0,T]\times \mathcal{X}^{n+1}
	\rightarrow \mathbb{R}$, and $\sigma^l: [0,T]\times \mathcal{X}^{n+1} \rightarrow (0,\infty)^{1\times m_1}$
	be continuous functions of the time variable for each $l = 1,\ldots,n$, where $m_1$ and $m_2$ are the positive integers. We also consider a collection of measurable functions
	$\eta_{lj}:\mathbb{R} \rightarrow (-1, \infty)$ for each $l=1,\ldots, n$, $j = 1, \ldots, m_2$.
	
	We further introduce some more notations. Fix $x=(x_0,x_1, \ldots, x_n)\in \mathcal{X}^{n+1}$ and
	$t\in [0,T]$ and we denote $b(t,x) := [\mu^1(t,x)-r(t,x), \ldots, \mu^n(t,x)-r(t,x)]_{1 \times n}$,
	and $\sigma(t, x) := [\sigma_{lj}(t, x)]_{n \times m_1}$, where $\sigma_{lj}$ is the $j$-th component
	of $\sigma^l$ function. For each $z = (z_1, \ldots, z_{m_2}) \in \mathbb{R}^{m_2}$, we denote $\eta(z)
	:= [\eta_{lj}(z_j)]_{n \times m_2}$. We use $[\cdot]^*$ to denote transpose of a vector.
	
	\subsection{Asset price model}
	Let $(\Omega, \mathscr{F}, P)$ be a  complete probability space. Let $\{X^l_0\mid l=0,\ldots, n\}$ be a collection of $(n+1)$ many $\mathcal{X}$ valued
	random variables, and  $\{Y^l_0\mid l=0,\ldots, n\}$ be a collection of $(n+1)$ non negative random variables.
	Let $W = \{W_t\}_{t \ge 0} = \{[W_t^1,\ldots,W_t^{m_1}]^{*}\}_{t \ge 0}$ be a standard $m_1$-dimensional Brownian motion. We further assume that, $\{N_j(dt,dz)|j=1,\ldots, m_2\}$  on $(0,\infty)\times \RR$ and $\{\wp^l(dt,dz_0) \mid l=0,\ldots, n\}$ on
	$(0,\infty)\times \RR$ are two sets of Poisson random measures  with intensities $\nu_j(dz)dt$ and $dtdz_0$ respectively defined on the same probability space. We recall that $\nu_j$ denotes a finite Borel measure for each $j$. It is important to note that the random variables, processes and measures are defined in such a way that they are independent. We denote the
	compensated measures by $\tilde{N}_j(dt,dz_j):= N_j(dt,dz_j) -
	\nu_j(dz_j) dt$ for $j=1,\ldots, m_2$ and $\tilde{\wp}^l(dt,dz_0):=\wp^l(dt,dz_0)-dtdz_0$ for $l=0,\ldots,n$. For each $l=0,1,\ldots,n$, let $\{X^l_t\}_{t\geq 0}$ be the solution to
	\eqref{agedep} with $\wp$ replaced by $\wp^l$, $\lambda$ by
	$\la^l$, $X_0$ by $X_0^l$, and $Y_0$ by $Y_0^l$. In other words
	\begin{eqnarray}
	X_t^l &=& X_0^l + \Int_{(0,t]}\Int_{\mathbb{R}} h^l(X^l_{u-}, Y^l_{u-},z_0)\wp^l(du,dz_0)\text{\label{1}}\\
	Y_t^l &=& Y_0^l+t- \Int_{(0,t]} \Int_{\mathbb{R}} g^l(X^l_{u-}, Y^l_{u-},z_0) \wp^l(du,dz_0)\text{\label{2}},
	\end{eqnarray}
	where $h^l:= h_{\la^l}$ and $g^l:= g_{\la^l}$. We denote the tuple $(X^0_t, X^1_t, \ldots, X^n_t)$
	by $X_t$ and $(Y^0_t, Y^1_t, \ldots, Y^n_t)$ by $Y_t$. Hence, $W$, $\{N_j(dt,dz),j=1,\ldots,m_2\}$
	and $X$ are independent. The process $\{Z^l_t\}_{t\geq 0}:=\{(X^l_t,Y^l_t)\}_{t\geq 0}$ is a time homogeneous strong Markov process.
	
	Let the  filtration $\{\mathscr{F}_t\}_{t\ge 0} $ be the right continuous augmentation of the filtration generated by $W,X,N_j~j=1,\ldots,m_2$ such that $\mathscr{F}_0$ contains all the $P$-null sets.
	We consider a frictionless market consisting of $(n+1)$ assets whose prices are denoted by $S_t^0, S_t^1, \ldots, S_t^{n-1}$ and $S_t^n$ and are traded continuously. We model the hypothetical state of the assets at time $t$ by the pure jump process $X=\{X_t\}_{t\ge 0}$. The state of the asset indicates its mean growth rate and volatility. We assume
	\begin{equation}
	\no dS_t^0 = r(t, X_t) S_t^0 dt, \quad S_0^0 = s_0 \geq 0.
	\end{equation}
	Thus the corresponding asset is (locally) risk free, which refers to the money market account with the floating interest rate  $r(t,x)$ at time $t$
	corresponding to regime $x$. The other $n$ asset prices are assumed to be given by the following stochastic differential equation
	\begin{align} \label{eq:sde}
	dS_t^l & = S_{t-}^l \left[ \mu^l(t, X_t )dt + \ds \sum_{j=1}^{m_1} \sigma_{lj}(t, X_t) ~dW_t^j +
	\ds \sum_{j=1}^{m_2} \Int_{\mathbb{R} } \!\eta_{lj}(z_j)\, N_j(dt, dz_j)\right],   \\
	S_0^l &= s_l, \quad s_l \geq 0, ~ l=1,2,\ldots, n.  \nonumber
	\end{align}
	These prices correspond to $n$ different risky assets.
	Therefore, $\mu^l$ represents the growth rate of the $l$-th asset
	and $\sigma$ the volatility matrix of the market. Here we further assume the following.
	
	\noi \textbf{Assumptions} :
	\begin{itemize}
		\item [\textbf{(A1)}]For each $l=1,\ldots, n$ and $j = 1, \ldots, m_2$, we assume
		$\eta_{lj}\in L^2(\nu_j)$.
		\item [\textbf{(A2)}]For each $l=1,\ldots, n$ and $j = 1, \ldots, m_2$, we further assume
		$\ln(1+\eta_{lj})\in L^2(\nu_j)$.
		\item[\textbf{(A3)}] Let $a(t,x) := \sigma(t,x) \sigma(t,x)^{*}$ denote the diffusion matrix. Assume that there exists a $\delta_1 > 0$ such that for each $t$ and $x$,
		$\xi^* a(t,x)\xi \geq \delta_1 \|\xi\|^2$, where $\|\cdot\|$ denotes the Euclidean norm.
	\end{itemize}
	
	\noi The next lemma asserts the existence and uniqueness of the solution to the  SDE (\ref{eq:sde}). The proof is deferred to the appendix.
	\begin{Lemma}\label{SDESol}
		Under the assumption \emph{\textbf{(A2)}} the equation (\ref{eq:sde}) has a  strong solution, which is adapted, a.s. unique and an rcll process.
	\end{Lemma}
	\begin{Remark}
		We note that \emph{\textbf{(A1)}} and \emph{\textbf{(A2)}} follow for the special case where
		\begin{align*}
		-1<\inf_{z\in\mathbb{R}}\eta_{lj}(z)\leq \sup_{z\in \mathbb{R}}\eta_{lj}(z)<\infty.
		\end{align*}
		By \emph{\textbf{(A3)}} the diffusion matrix $a(t,x)$ is uniformly positive
		definite, which ensures that $a(t,x)$ is invertible. We will use this
		condition in Section 3. This condition also implies that $m_1 \ge
		n$.
	\end{Remark}
	
	\subsection{Portfolio value process}
	Consider an investor who is employing a self-financing portfolio of the above $(n+1)$ assets starting
	with a positive wealth. If the portfolio at time $t$ comprises of $\pi_t^l$ number of units of $l$-th
	asset for every $l=0,\ldots, n$, then for each $\omega \in \Omega$ the value of the portfolio at time $t$ is given by
	$$ V_t := \ds\sum_{l=0}^{n} \pi_t^l ~S_t^l. $$
	We allow $\pi_t^l$ be real valued, i.e., borrowing from the money market and short selling of assets are allowed. We further assume that $\{\pi_t^l\}_{t\ge 0}$ is an $\{\mathscr{F}_t\}_{t\ge 0}$ adapted, rcll process for each $l$. Then the self-financing condition implies that
	$$ dV_t ~=~ \ds\sum_{l=0}^{n} \pi_{t-}^l ~dS_t^l.$$
	If $\pi_t^l$ are such that $V_t$ remains positive, we can set $u_t^l := \frac{\pi_t^l S^l(t)}{V_{t}}$, the
	fraction of investment in the $l$-th asset. Then we have ${\sum_{l=0}^{n}} u_t^l = 1$ and hence $u_t^0 = 1 - {\sum_{l=1}^{n}} u_t^l$.
	We call $u_t = [u_t^1, u_t^2, \ldots, u_t^n]^{*}$ as the portfolio strategy of risky assets at time $t$. Then the wealth process, $\{V_t\}_{t\geq 0}$, now onward denoted by $V^u:=\{V^u_t\}_{t\geq 0}$, takes the form
	$$\frac{dV_t^{u}}{V_{t-}^{u}} = \ds\sum_{l=0}^{n} u_{t-}^l \frac{dS_t^l}{S_{t-}^l}. $$
	Thus we would consider the following SDE for the value process,
	\begin{align} \label{eq:pov}
	\no dV_t^{u} & =  V_t^{u} \left(~r(t,X_t) + \ds\sum_{l=1}^{n} \left[\mu^l(t, X_t) - r(t,X_t)\right] u_t^l\right) dt \\
	\no & \quad\quad + V_t^{u} \ds\sum_{l=1}^{n} \ds\sum_{j=1}^{m_1} \sigma_{lj}(t, X_t) ~u_t^l dW_t^j  +  V_{t-}^{u} \ds\sum_{l=1}^{n} \ds\sum_{j=1}^{m_2}u_{t-}^l \Int_{\mathbb{R}} \eta_{lj}(z_j) N_j(dt, dz_j) \\
\no &= V_t^{u} (r(t,X_t) + b(t,X_t)u_t) dt  + V_t^{u} u_t^{*}\sigma(t, X_t)dW_t \\
	 & \quad\quad   + V_{t-}^{u}  \ds\sum_{j=1}^{m_2} \int_{\mathbb{R}} \left[u_{t-}^{*} \eta(z)\right]_j N_j(dt, dz_j),
	\end{align}
	\noi where $u_t^{*} \eta(z) = \left[ \sum_{l=1}^{n}u_t^l\eta_{l1}(z_1), \ldots,
	\sum_{l=1}^{n}u_t^l\eta_{lm_2}(z_{m_2}) \right]_{1 \times m_2}$. Note that, some additional assumptions on $u$ are needed for ensuring a positive strong solution of \eqref{eq:pov}.
	
	\begin{Remark}\label{rem1}
		As before, we need to assume that $u_t$ is such that for each
		$j=1,\ldots, m_2$, and $z\in \RR^{m_2}$, $\left[u_{t-}^{*}
		\eta(z)\right]_j > -1$ to ensure a positive solution to
		\eqref{eq:pov}. For some technical reasons we require a stronger condition on $u_t$. We would require that the the portfolio should be chosen from
		\begin{align}\label{udelta}
		\mathcal{U}_\delta:=\{u\in\mathbb{R}^n|\left[u^{*} \eta(z)\right]_j
		\geq -1+\delta, \forall j, z\} ~\quad \text{for some}~ 0<\delta\leq 1.
		\end{align}
	\end{Remark}
	It is clear from the definition and  the above derivation that $V^{u}$, the portfolio wealth process, is a  controlled  process. Let $\mathbb{A} \subseteq \mathbb{R}^n$ be a convex set containing the origin, denoting the range of portfolio.
	The range is determined based on the investment restrictions. For example, $\mathbb{A} = \mathbb{R}^n$ in the case of unrestricted short selling. The restrictions on short selling makes $\mathbb{A} = \{ u\in \mathbb{R}^n \mid u^l\ge c_l, \sum_{l\ge 1} u^l \le 1-c_0 \forall l\}$, where $c_l\le 0$ for $l=0, \ldots, n$. Clearly, $c_l=0$ for $l=0, \ldots, n$, correspond to no short selling.
	\begin{Definition}\label{defi1}
		An rcll and adapted process $u=\{u_t\}_{t\in [0,T]}$ is said to be \textit{admissible} portfolio strategy if:
		\begin{itemize}
			\item[(i)]  the process $u$ takes values from the convex set $\mathbb{A}_1:= \mathbb{A} \bigcap \mathcal{U}_\delta$, where $ \mathcal{U}_\delta$ is as in  \eqref{udelta},
			\item[(ii)] \eqref{eq:pov} has an almost sure unique strong solution,
			\item[(iii)] $\emph{ess}\displaystyle\sup_\Omega \sup_{[0,T]}\|u_t(\omega)\|<\infty$.
			
		\end{itemize}
		
	\end{Definition}

	\begin{Proposition}\label{solstate}
		Under \textbf{(A1)} and with admissible control $u$, (i) the SDE \eqref{eq:pov} has an almost sure unique positive strong solution,  
		(ii) the solution has finite moments of all positive and negative orders, which are also bounded on $[0,T]$ uniformly in $u$.
	\end{Proposition}
	
	\proof (i) We first note that, since $u_t\in\mathcal{U}_\delta$ and satisfies Definition \ref{defi1}(iii),
	$$|\ln(1+[u_{s-}^*\eta(z)]_j)|<\max\left(|\ln \delta|, C\|\eta_{\cdot j}(z_j)\|\right),$$  where $C:=\emph{ess}\displaystyle\sup_\Omega \sup_{[0,T]}\|u_t(\omega)\|$ and $\eta_{\cdot j}$ is the $j$-th column of the matrix $\eta$.
	Again using \textbf{(A1)} and the finiteness of the measure $\nu_j$, the integration of the above upper bound with respect to $N_j$ has finite expectation. This implies that $\mathbb{E}\Int_{0}^{t  }\! \Int_{\mathbb{R}}\ln(1+ [u_{s-}^*\eta(z)]_j)\,N_j(ds, dz_j)<\infty$. Therefore in the similar line of the proof of Lemma
	\ref{SDESol}, we can show, under the assumption \textbf{(A1)} and the admissibility of $u$, (\ref{eq:pov}) has an a.s. unique positive
	rcll solution, which is an adapted process, and the solution is given by
	\begin{align}\label{SolforV}
	V_{t}^u & = V_0^u\exp\Bigg[\Int_{0}^{t  } \!\Bigg(r(s, X_s)+b(s, X_s)u_s - \frac{1}{2} u_s^*a(s,X_s)u_s\Bigg)ds\left.+\Int_{0}^{t  }\! u_s^*\sigma(s, X_s)\,dW_s\right.\nonumber \\
	& \quad  +\ds\sum_{j=1}^{m_2}\Int_{0}^{t  }\! \Int_{\mathbb{R}}\ln(1+ [u_{s-}^*\eta(z)]_j)\,N_j(ds, dz_j) \Bigg].
	\end{align}
	(ii) We first consider the first order moment. To prove for each $t$, $ V_{t}^u$ has a bounded expectation, we first note that the right hand side can be written as a product of a conditionally log-normal random variable and $\exp\left(\ds\sum_{j=1}^{m_2}\Int_{0}^{t  }\! \Int_{\mathbb{R}}\ln(1+ [u_{s-}^*\eta(z)]_j)\,N_j(ds, dz_j)\right)$, where both are conditionally independent, given the process $u$. We further note that the log-normal random variable has bounded parameters on $[0,T]$ uniformly in $u$. Therefore  it is sufficient to check if
	\begin{align*}
	\mathbb{E}\left[\exp\left(\Int_{0}^{t  }\! \Int_{\mathbb{R}}\ln(1+ C\|\eta_{\cdot j}(z_j)\|)\,N_j(ds, dz_j)\right)\right],
	\end{align*}
	is bounded on $[0,T]$, for all $j=1,\ldots,m_2$.  By applying Lemma \ref{expeta1}, one can show that the above expectation is bounded on $[0,T]$.  Thus $ V_{t}^u $ has bounded expectation on $[0,T]$, uniformly in $u$. Now for the moments of general order, we note that for any $\alpha\in \mathbb{R},~ (V^u_t)^\alpha$ can also be written in a similar form of \eqref{SolforV} where each of the integrals inside the exponential would be multiplied by the constant $\alpha$. Thus the rest of the proof follows in a similar line of that of first order case, given above.\qed
	
	Our goal is to study a risk sensitive optimal control problem on the above wealth process. We would see in the next section that, in order to obtain a classical solution to the corresponding HJB equation, to be defined shortly, certain regularity of the conditional c.d.f of holding time of $X$ is needed. We devote the next subsection to establishing some smoothness of relevant density functions.
	
	\subsection{Regularity properties of holding time distributions}
Let $T^l_n$ be the time of $n$-th
transition of the $l$-th component of $X_t$, whereas $T^l_0=-Y^l_0$ and
$\tau^l_n:=T^l_n-T^l_{n-1}$. We define the function $F^l:[0, \infty)\times \mathcal{X}\rightarrow [0,1]$ as
	$F^l(\bar{y}|i):=1-e^{-\Lambda^l_i(\bar{y})}$ and let
	$f^l(\bar{y}|i):=\frac{d}{d\bar{y}}F^l(\bar{y}|i)$ and for each
	$i\neq j,
	p^l_{ij}(\bar{y}):=\frac{\la^l_{ij}(\bar{y})}{|\la^l_{ii}(\bar{y})|}$
	with $p^l_{ii}(\bar{y})=0$ for all $i$ and $\bar{y}$. Set
	$$\hat{p}^l_{ij}=\Int_0^\infty p^l_{ij}(\bar{y})
	dF^l(\bar{y}|i).$$
	We assume further conditions on the transition rate so that the unconditional transition probability matrix is irreducible.\\
	\noi \textbf{Assumption:} \textbf{(A4)} The matrix $(\hat{p}^l_{ij})$ is irreducible,
	for all $l=0,\ldots,n.$

	\noi From the definition of $F^l$ and the
	assumptions on $\la$, we observe $F^l(\bar{y}|i)<1$, for all
	$\bar{y}>0$. We also note that $\la^l_{ij}(\bar{y})
	=p^l_{ij}(\bar{y})\frac{f^l(\bar{y}|i)}{1-F^l(\bar{y}|i)}$ hold
	for all $i\neq j$.  For a fixed $t$, let
	$n^l(t):=\max\{n:T^l_n\leq t\}$. Hence $T^l_{n^l(t)}\leq t \leq
	T^l_{n^l(t)+1}$ and $Y^l_t=t-T^l_{n^l(t)}$. It is shown in
	\cite{GhS} that $\la^l:\mathcal{X}_2\times[0,\infty)\to (0,\infty)$ is the instantaneous transition rate function of the semi-Markov process $X^l$,  i.e.,
	$$\la^l_{ij}(\bar{y}) = \lim_{\delta\to 0} \frac{1}{\delta}P\left( X^l_{T^l_{n+1}}  = j, \tau^l_{n+1}  \in (\bar{y}, \bar{y} + \delta)|X^l_{T^l_ n}  = i, \tau^l_n > \bar{y}\right).$$ Furthermore,  $F^l(\bar{y}|i)$ is the conditional c.d.f of the
	holding time of $X^l$ and $p^l_{ij}(\bar{y})$ is the conditional
	probability that $X^l$ transits to $j$ from $i$ given the fact that it is
	at $i$ for a duration of $\bar{y}$. Let $\tau^l(t):=$ the remaining life of $l$-th component i.e., the time period from time $t$
	after which the $l$-th component of $X$ would have the subsequent transition.
	Note that $\tau^l(t)$ is independent of every component of $X$
	other than $l$-th one. We denote the conditional c.d.f and p.d.f of $\tau^l(t)$
	given $X^l_t=i$ and $Y^l_t=\bar{y}$ as
	$F_{\tau^l}(\cdot|i,\bar{y})$ and $f_{\tau^l}(\cdot|i,\bar{y})$ respectively. It is important to note that this
	c.d.f does not depend on $t$, mainly because $(X_t,Y_t)$ is time-homogeneous. We also notice that $\tau^l(t)+Y^l_t$ is the duration
	of stagnancy of $X^l_t$ at present state before it moves to another.
	From now we denote $P(\cdot|X_t=x,Y_t=y)$ by $P_{t,x,y}(\cdot)$
	and the corresponding conditional expectation as
	$\mathbb{E}_{t,x,y}(\cdot)$. Let $\ell(t)$ be the component of $X_t$, where the subsequent transition happens.
	Therefore, $P_{t,x,y}(\ell(t)=l)$ represents the conditional probability of observing next
	transition to occur at the $l$-th component given that $X_t=x$ and $Y_t=y$. We find the expressions of the c.d.f and the probability defined above and obtain some properties in the following lemma. The proof is deferred to the appendix. In order to state the lemma, we introduce some more notations. We define an open set
	$$\mathscr{D}:= \{ (t,x,y)\in (0,T)\times\mathcal{X}^{n+1}\times (0,\infty)^n\},$$
	and a linear operator $$D_{t,y}\vf(t,x,y):=\displaystyle\lim_{\vp\rightarrow 0}\frac{1}{\vp}\{\vf(t+\vp,x,y+\vp\mathbf{1})-
	\vf(t,x,y)\},$$
	where dom($D_{t,y}$), the domain of $D_{t,y}$ is the subspace of $C(\mathscr{D})$ such that for each $\vf \in$ dom($D_{t,y}$) above limit exists for every $(t,x,y)\in \mathscr{D}$ and $D_{t,y}\vf \in C(\mathscr{D})$, and $\mathbf{1}\in \RR^{(n+1)}$ is a vector with each component $1$.
	\begin{Lemma}\label{theo1}
		Consider $F^l,f^l,P_{t,x,y}$ as given above.\\
		\noi (i) For each $l$,
		\begin{eqnarray*}
			P_{t,x,y}(\ell(t)=l)=\Int_0^\infty\prod_{m\neq l}\frac{1-F^m(s+y^m|x^m)}
			{1-F^m(y^m|x^m)}\frac{f^l(s+y^l|x^l)}{1-F^l(y^l|x^l)}ds.
		\end{eqnarray*}
		\noi (ii) Let $F_{\tau^l|l}(\cdot|x,y)$ be the conditional c.d.f of $\tau^l(t)$ given $X_t=x, Y_t=y$ and $\ell(t)=l$. Then
		\begin{eqnarray}\label{Ftaullv}
		F_{\tau^l|l}(r|x,y)= \frac{\Int_0^r\prod_{m\neq l}(1-F^m(s+y^m|x^m))f^l(s+y^l|x^l)ds}
		{\Int_0^\infty\prod_{m\neq l}(1-F^m(s+y^m|x^m))f^l(s+y^l|x^l)ds},
		\end{eqnarray}
		and is $C^2$ in $r$ variable.
		
		\noi (iii) \begin{eqnarray}\label{exftaul}
		f_{\tau^l|l}(r|x,y):= \frac{d}{dr}F_{\tau^l|l}(r|x,y)=\frac{\prod_{m\neq l}
			(1-F^m(r+y^m|x^m))f^l(r+y^l|x^l)}{\Int_0^\infty\prod_{m\neq l}(1-F^m(s+y^m|x^m))f^l(s+y^l|x^l)ds},
		\end{eqnarray}
		is differentiable with respect to $r$.
		
		\noi (iv) $F_{\tau^l|l}(T-t|x,y)$ and $P_{t,x,y}(\ell(t)=l)$ are  in \emph {dom}($D_{t,y}$) . Furthermore,
		\begin{align*}
		D_{t,y}P_{t,x,y}(\ell(t)=l) = &\ds\sum_{m=0}^{n} f_{\tau^m}(0|x^m,y^m)P_{t,x,y}(\ell(t)=l)-f_{\tau^l}(0|x^l,y^l)\\
		D_{t,y}F_{\tau^l|l}(T-t|x,y)= &f_{\tau^l|l}(0|x,y) (F_{\tau^l|l}(T-t|x,y)-1).
		\end{align*}
		
		\noi (v) $f_{\tau^l|l}(0|x,y)P_{t,x,y}(\ell(t)=l)=\frac{f^l(y^l|x^l)}{1-F^l(y^l|x^l)}=f_{\tau^l}(0|x^l,y^l)$.
		
	\end{Lemma}

	\subsection{Optimal Control Problem} \label{sec:HJB}

	\noi In this paper we consider a risk sensitive optimization criterion of the terminal portfolio wealth corresponding to a portfolio $u$, that is given by
	\begin{align}\label{eq:riskSens1}
	J_\theta^{u,T}(x,y,v) & := -\left(\frac{2}{\theta}\right) ~\ln\mathbb{E}\left[\exp\left(-\frac{\theta}{2}
	\ln\left(V_T^u\right)\right) ~\biggr|~ X_0=x, Y_0=y, V^u_0=v\right] \nonumber \\
	\nonumber
	&  = -\left(\frac{2}{\theta}\right) ~\ln\mathbb{E}\left[ (V_T^u)^{-\frac{\theta}{2}} ~\biggr|~ X_0=x, Y_0=y, V^u_0=v\right],
	\end{align}
	which is to be maximized over all admissible portfolio strategies with constant risk aversion parameter $\theta > 0$. Since logarithm is increasing, it suffices
	to consider the following cost function
	\begin{equation}\label{eq:riskSens2}
	\no \mathbb{E}\left[(V_T^u)^{-\frac{\theta}{2}} ~\biggr|~ X_0=x, Y_0=y, V^u_0=v\right],
	\end{equation}
	which is to be minimized.
	For all $ (t,x,y,v)\in \mathscr{D}\times (0,\infty)$, let
	\begin{equation}\label{expressvf}
	\left.\begin{array}{ll}
	\tilde{J}_\theta^{u,T}(t,x,y,v) :=  \mathbb{E}\left[(V_T^u)^{-\frac{\theta}{2}} ~\biggr|~ X_t=x, Y_t=y, V^u_t=v\right],\\
	\vf_\theta(t,x,y,v) := \inf_{u}\tilde{J}_\theta^{u,T}(t,x,y,v),
	\end{array}\right\}
	\end{equation}
	\noi where the infimum is taken over all admissible strategies as in Definition \ref{defi1}. Hence, $\vf_\theta$ corresponds to the optimal value.

	Let $u=\{u_t\}_{t\in [0,T]}$ be an admissible strategy such that it has the following form $u_t:=\tilde u(t,X_t,Y_t,V_t)$ for some measurable $\tilde u:\mathscr{D}\times(0,\infty)\to \mathbb{A}_1$. We call such controls as Markov feedback control. Then the augmented process  $\{( X_t, Y_t,V_t^u)\}_{t\in [0,T]}$ is Markov where, $X_t, Y_t,V_t^u$ are as in \eqref{1}, \eqref{2}, (\ref{eq:pov}). We note that for any measurable $\tilde u:\mathscr{D}\times(0,\infty)\to \mathbb{A}_1$, the equation \eqref{eq:pov} may not have a strong solution. However, we will show the existence of a Markov feedback control which is optimal and under which \eqref{eq:pov} has an a.s. unique strong solution.
	
	Let $\mathscr{A}^{\tilde u} $ be the infinitesimal generator  of $\{(t,X_t, Y_t,V_t^u)\}_{t\in [0,T]}$, and $\vf$ be a $C^\infty$ function with compact support, then we have
	\begin{align}\label{infGen}
	& \mathscr{A}^{\tilde u} \vf(t, x, y,v) \nonumber\\
	&= D_{t,y}\vf(t, x, y,v)+v \left[r(t,x) + b(t,x)~\tilde u(t,x,y,v)\right] \frac{\partial }{\partial v} \vf(t, x, y,v)
	\nonumber\\
	& \quad + \frac{1}{2}  v^2 \left[\tilde u^{*}(t,x,y,v)a(t, x)\tilde u(t,x,y,v) \right]\frac{\partial^2 }{\partial v^2}\vf(t, x, y,v)\nonumber \\
	& \quad + \ds\sum_{j=1}^{m_2} \Int_{\mathbb{R}} \!\left[\vf\left(t, x, y,v\left(1 + [\tilde u^{*}(t,x,y,v)  \eta(z)]_j\right)\right)
	- \vf(t, x, y,v) \right]\, \nu_j(dz_j) \nonumber \\
	& \quad +\ds\sum_{l=0}^{n}\ds\sum_{j\neq x^l}\lambda^l_{x^lj}(y^l)\left[\vf(t,R^l_jx,R^l_0y,v) - \vf(t, x, y,v)\right],
	\end{align}
	where the linear operator $R^l_z$ is given by $R^l_zx:=x+(z-x^l)e_l$, $l=0,\ldots,n$, $z\in \RR$ and $\{e_l : l=0,\ldots,n\}$ is the standard basis of $\mathbb{R}^{n+1}$.  For a given $u \in \mathbb{A}_1$, by abuse of notation, we write $\mathscr{A}^u$, when $\tilde u(t,x,y,v)=u$ for all $t,x,y,v$. We consider the following HJB equation
	\begin{align}\label{hjbRS}
	\inf_{u \in \mathbb{A}_1}\mathscr{A}^u\vf(t,x,y,v) = 0,
	\end{align} with the terminal condition
	\begin{align}\label{hjbConsRS}
	\vf(T,x,y,v) = v^{-\frac{\theta}{2}}, ~  x\in \mathcal{X}^{n+1}, ~ y\in[0,T]^{n+1}, \quad v>0.
	\end{align}
	To study the HJB equation we now define following classes of functions
	$$\mathcal{V}:=\{\psi\in C\left((0,\infty)\right)|\displaystyle\sup_{v\in (0,\infty)}{|v^{\frac{\theta}{2}}\psi(v)}|<\infty\}.$$
	
	\begin{Definition}\label{defiG}
		Let $\mathscr{G}\subset\{\vf : \mathscr{D}\times(0,\infty)
		\to \mathbb{R}\}$ be such that for every $\vf \in \mathscr{G}$ the following hold:
		\begin{itemize}
			
			\item[(i)] $\vf(t,x,y,v)$ is twice continuously differentiable with respect to $v\in (0,\infty)$ for all $t\in (0,T),x\in \mathcal{X}^{n+1},y\in (0,t)^{n+1} $ and
			$\vf$ is in dom($D_{t,y}$) for each $v$,$x$,
			\item[(ii)] for fixed $(t,x,y)\in\mathscr{D}$, $\vf(t,x,y,\cdot)\in \mathcal{V}$,
			\item[(iii)]for each $(t,x,y)$,  $v\mapsto v\frac{\partial \vf}{\partial v}$ is in $\mathcal{V}$.
		\end{itemize}
	\end{Definition}
	We now define a classical solution to the problem \eqref{hjbRS}-\eqref{hjbConsRS}.
	\begin{Definition}\label{defcl}
		We say $\varphi:\mathscr{D}\times (0,\infty)\rightarrow \mathbb{R}$ is a classical solution to \eqref{hjbRS}-\eqref{hjbConsRS} if
		$\vf \in \mathscr{G}$ and for all $(t,x,y,v)\in \mathscr{D}\times (0,\infty)$, $\varphi$ satisfies \eqref{hjbRS}-\eqref{hjbConsRS}.
	\end{Definition}
	
\section{Hamilton-Jacobi-Bellman Equation} \label{sec:RskSens}
	\noi We look for a solution to \eqref{hjbRS}-\eqref{hjbConsRS} of the form
	\begin{equation}\label{eq:RSTransf}
	\vf(t,x,y,v) = v^{-\frac{\theta}{2}}\psi(t,x,y),
	\end{equation}
	where $\psi \in \text{dom}(D_{t,y})$. Clearly, the left hand side of \eqref{eq:RSTransf} is in class $\mathscr{G}$.  We will establish the following result in first two subsections.
	
	\begin{Theorem}\label{theo2}
		The Cauchy problem \eqref{hjbRS}-\eqref{hjbConsRS} has a unique classical solution, $\vf_M$, of the form \eqref{eq:RSTransf}.
	\end{Theorem}
	Substitution of (\ref{eq:RSTransf}) into \eqref{hjbRS}, yields
	\begin{align}\label{eq:hjbPdeRS}
	D_{t,y}\psi(t,x,y)+\ds\sum_l\ds\sum_{j\neq x^l}\la^l_{x^lj}(y^l)\Big[\psi(t,R^l_jx,R^l_0y)-\psi(t,x,y)\Big]
	+h_\theta(t,x)\psi(t,x,y)=0,
	\end{align}
	for each $(t,x,y)\in \mathscr{D}$ with the condition
	\begin{equation}\label{eq:hjbPdeTerRS}
	\psi(T,x, y) = 1,
	\end{equation}
	where the map $h_{\theta}:[0,T]\times \mathcal{X}^{n+1}\rightarrow \mathbb{R}$ is given by
	\begin{equation}\label{htheta}
	h_\theta(t,x) := \inf_{u \in \mathbb{A}_1}\left[g_\theta(t,x,u)\right],
	\end{equation}
	the infimum of a family of continuous functions
	\begin{align}
	g_\theta(t,x,u) & := \left(-\frac{\theta}{2}\right) \left[r(t,x) + b(t,x)~u\right]  + \frac{1}{2}  \left(-\frac{\theta}{2}\right)\left(-\frac{\theta}{2}-1\right)\left[u^{*}~a(t,x)~u \right]  \nonumber \\
	&  \quad + \ds\sum_{j=1}^{m_2} \Int_{\mathbb{R}} \!\left(\left(1 + [u^{*}  \eta(z)]_j\right)^{\left(-\frac{\theta}{2}\right)}- 1 \right)\, \nu_j(dz_j).\nonumber
	\end{align}
	It is important to note that the linear first order equation \eqref{eq:hjbPdeRS} is nonlocal due to the presence of the term $\psi(t,R^l_jx,R^l_0y)$ in the equation. It implies that $D_{t,y} \psi(t,x,y)$ depends on the value of $\psi$ at the point $(t,\cdot,R^l_0y)$, which does not lie in the neighbourhood of $(t,\cdot, y)$. We now define a classical solution to \eqref{eq:hjbPdeRS}-\eqref{eq:hjbPdeTerRS} below.
	\begin{Definition}\label{deffcl}
		We say $\varphi:\mathscr{D}\rightarrow \mathbb{R}$ is a classical solution to \eqref{eq:hjbPdeRS}-\eqref{eq:hjbPdeTerRS} if $\varphi\in$ dom$(D_{t,y})$ and for all $(t,x,y)\in \mathscr{D}$, $\varphi$ satisfies \eqref{eq:hjbPdeRS}-\eqref{eq:hjbPdeTerRS}.
	\end{Definition}
	\begin{Remark}
It is interesting to note that other than the terminal condition \eqref{eq:hjbPdeTerRS}, no additional boundary conditions are imposed.  The remaining part of the boundary is $\bar{\mathscr{D}}\cap \{(t,x,y)|y^l =0, x \in \mathcal{X}^{n+1}, t\in [0,T]\}$. We note from \eqref{2} that, $0\leq Y^l_t$, for all $t\in [0,T]$. Hence $\{Y_t\}_{t\geq 0}$ does not cross the boundary. Thus the value of solution on the boundary is obtained from the terminal condition \eqref{eq:hjbPdeTerRS}.
	\end{Remark}
	\begin{Theorem}\label{theo3}
		The Cauchy problem \eqref{eq:hjbPdeRS}-\eqref{eq:hjbPdeTerRS} has a unique classical solution in $C_b(\bar{\mathscr{D}})$.
	\end{Theorem}
	
	\begin{Remark}
		Note that Theorem \ref{theo2} may be treated as a corollary of Theorem \ref{theo3} in view of the substitution (\ref{eq:RSTransf}) and subsequent analysis. Thus it suffices to establish Theorem \ref{theo3}. We establish Theorem \ref{theo3} in the subsection \ref{lfoe} via a study of an integral equation which is presented in subsection \ref{tvie}. The following result would be useful to establish well-posedness of \eqref{eq:hjbPdeRS}-\eqref{eq:hjbPdeTerRS}.
	\end{Remark}
	
	\begin{Proposition}\label{hContinuityRS}
		Consider the map $h_{\theta}:[0,T]\times \mathcal{X}^{n+1}\rightarrow \mathbb{R}$, given by, \eqref{htheta}. Then under \emph{\textbf{(A3)}}, we have
		\begin{itemize}
			\item[(i)] $h_\theta$ is continuous, negative valued and bounded below;
			\item[(ii)] $H_\theta(t_1,t_2,x) :=\int_{t_1}^{t_2}h_\theta(s,x)ds$ is $C^1$ in both $t_1$ and $t_2$ for each $x$;
			\item[(iii)] For every $(t,x)$, there exists a unique $u^{\star}(t,x) \in \mathbb{A}_1$ such that $h_{\theta}(t,x) = g_{\theta}(t,x,u^{\star}(t,x))$  and $u^{\star}:[0,T]\times \mathcal{X}^{n+1}\to \mathbb{A}_1$ is continuous in $t$;
			\item[(iv)]$\bar{u}^\star:=\{u^{\star}(t,X_t)\}_{t\geq 0}$ is admissible.
		\end{itemize}
	\end{Proposition}
	\begin{proof}
		(i) We recall that, ${\mathbb A}_1$, the range of portfolio includes the origin. Therefore
		\begin{eqnarray*}
			h_\theta (t,x) & \leq & g_\theta(t,x,0)= -\frac{\theta}{2}r(t,x)<0.
		\end{eqnarray*}
		Thus $h_\theta$ is negative valued. By the continuity assumptions on $r,b$ and $a$, for fixed $u$ and each $x \in \mathcal{X}^{n+1}$, $r(t,x)$, $b(t,x),$ and $a(t,x)$ are bounded on $[0,T]$. Let $M \geq  0$ be such that $$ \max_{t \in [0,T]}\{|r(t,x)|, \|b(t,x)\|, \|a(t,x)\|\} \leq M.$$
		We also observe that for each $u\in \mathbb{A}_1$,
		\begin{eqnarray*}
			\ds\sum_j \Int_{\mathbb{R}} (\left(1 + [u^{*}  \eta(z)]_j\right)^{-\frac{\theta}{2}}-1)\, \nu_j(dz_j) &\geq&  -\ds\sum_j \Int_{\mathbb{R}} \nu_j(dz_j) \\
			& = & -\ds\sum_j\nu_j(\mathbb{R})>-\infty,
		\end{eqnarray*}
		using the finiteness of the measure $\nu_j$.
		Also, \textbf{(A3)} gives $u^*a(t,x)u \geq \delta_1 \|u\|^2.$  Hence by using the above mentioned bounds, we can write, $g_\theta(t,x,u)  \geq \bar{g}_\theta(u)$,  where
		\begin{align*}
		\bar{g}_\theta(u) & = \left(-\frac{\theta}{2}(M + M\|u\|) + \frac{\theta}{4}(1+\frac{\theta}{2})\delta_1 \|u\|^2 -\ds\sum_j\nu_j(\mathbb{R}) \right).
		\end{align*}
		Since $\bar{g}_\theta(u)$ is independent of $t$ and $\uparrow\infty$ as $\|u\|\uparrow \infty$,  $h_\theta (t,x)$ is bounded below. Now we will show that for fixed $t$ and $x$, $g_\theta(t,x,u)$ is a strictly convex function of variable $u \in \mathbb{A}_1$. For fixed $t$ and $x$, let $H$ denote the Hessian matrix for $g_\theta$. Then $(p,q)$-th element of $H$,
		\begin{align*}
		H_{pq}=&\frac{\partial^2 g_\theta}{\partial u_p \partial u_q}\\
		=&\frac{1}{2}\frac{\theta}{2}\left(\frac{\theta}{2}+1\right)a_{pq}(t,x)+\ds\sum_j \Int_{\mathbb{R}} \frac{\theta}{2}\left(\frac{\theta}{2}+1\right)\eta_{pj}(z_j)\eta_{qj}(z_j)\left(1 + [u^{*}  \eta(z)]_j\right)^{-\frac{\theta}{2}-2}\, \nu_j(dz_j).
		\end{align*}
		Since $u$ is in $\mathbb{A}_1$, $(1 + [u^{*}  \eta(z)]_j$ is bounded below by a positive $\delta$. Hence, in addition to that using \textbf{(A3)}, there exists $m>0$ such that $H-mI$ is a positive definite matrix and this proves the strict
		convexity of $g_\theta(t,x,u)$ on variable $u$. Therefore $ \mathbb{A}_2:=\mathbb{A}_1\bigcap {\bar{g}_{\theta}}^{-1} \left((-\infty,1]\right)$ is a non-empty convex compact set. Hence,  $(t,x)\twoheadrightarrow \mathbb{A}_2$ is a compact-valued correspondence. Since $h_\theta$ is negative, from \eqref{htheta}, we can write
		\begin{eqnarray*}
			h_\theta(t,x)&=&\inf\{g_\theta(t,x,u)|u\in\mathbb{A}_2 \}.
		\end{eqnarray*}
		We also note that $(t,x,u)\mapsto g_\theta(t,x,u)$ is jointly continuous.
		Since $(t,x)\twoheadrightarrow \mathbb{A}_2$ is continuous, then it follows from the Maximum Theorem (\cite{Rang},Th. $9.14$) that $h_{\theta}(t,x)$ is continuous with respect to $(t,x)$.
		Hence (i) is proved.
		
		(ii) Follows from the continuity of $h_{\theta}(t,x)$.
		
		(iii) The set of minimizers is defined by
		\begin{align*}
		u^\star(t,x)&= \text{argmin}\{g_\theta(t,x,u)|u\in\mathbb{A}_2\}.
		\end{align*}
		Again by using (\cite{Rang},Th. $9.14$), $(t,x)\twoheadrightarrow u^{\star}(t,x)$ is upper semi-continuous. Since $g_\theta(t,x,u)$ is strictly convex in $u$, for each $t \in [0,T]$ and $x \in \mathcal{X}^{n+1}$ there exist only one element in $u^{\star}(t,x)$. By abuse of notation, we denote that element by $u^{\star}(t,x)$ itself. Since a single-valued upper semi-continuous correspondence is continuous, $u^{\star}(t,x)$ is a continuous function.
		
		(iv) Since $u^\star$ is continuous in $t$, there exists a positive constant $M$ such that
		$ \|u^\star(t, x)\|<M$ for all $t\in [0,T]$, $x\in \mathcal{X}^{n+1} $. Thus $\bar{u}^\star$ is bounded. Since $\bar{u}^\star$ does not depend on $v$, the Lipschitz conditions of Theorem 1.19 of \cite{OkSu} are satisfied. Again since $\bar{u}^\star$ is bounded, all growth conditions are also satisfied. Therefore Definition \ref{defi1}(ii) is satisfied and this completes the proof.
	\end{proof}
	
	\subsection{Volterra Integral equation}\label{tvie}
	In order to study \eqref{eq:hjbPdeRS}-\eqref{eq:hjbPdeTerRS} we consider the following integral equation with the previous notations and for all $(t,x,y)\in \bar{\mathscr{D}}$
	\begin{align}\label{intrisk}
	\psi(t,x,y)&=\ds\sum_{l=0}^{n}P_{t,x,y}(\ell(t)=l)\Big[(1-F_{\tau^l|l}(T-t|x,y))e^{H_\theta(t,T,x)} +\int_{0}^{T-t}e^{H_\theta(t,t+r,x)}\ds\times \nonumber \\
	& \sum_{j\neq x^l}p^l_{x^lj}(y^l+r)\psi(t+r,R^l_jx,R^l_0(y+r\mathbf{1}))f_{\tau^l|l}(r|x,y)dr \Big].
	\end{align}
	\vskip .2in
	\noi Equation \eqref{intrisk} is a Volterra integral equation of second kind. We note that the boundary of $\bar{\mathscr{D}}$ has many facets. For $t=T$, we directly obtain from \eqref{intrisk}, $\psi(T,x,y)=1$. Hence no additional terminal conditions are required. Although the values of $\psi$ in facets $\bar{\mathscr{D}}\cap \{(t,x,y)|y^l =0,x \in \mathcal{X}^{n+1}, t\in [0,T]\}$ are not directly followed but can be obtained by solving the integral equation on the facets.
	\begin{Proposition}\label{IESolnRS}
		(i) The integral equation \eqref{intrisk} has a unique solution in $C_b(\bar{\mathscr{D}})$, and (ii) the solution is in the \emph{dom}($D_{t,y}$).
	\end{Proposition}
	\proof (i) We first observe that the solution to the integral equation \eqref{intrisk} is a fixed point of the operator $A$, where
	\begin{align*}
	\no A\psi(t,x,y)&:=\ds\sum_{l=0}^{n}P_{t,x,y}(\ell(t)=l)\Big[(1-F_{\tau^l|l}(T-t|x,y))e^{H_\theta(t,T,x)}\\
	&+\int_{0}^{T-t}e^{H_\theta(t,t+r,x)}\ds\sum_{j\neq x^l}p^l_{x^lj}(y^l+r)\psi(t+r,R^l_jx,R^l_0(y+r\mathbf{1}))f_{\tau^l|l}(r|x,y)\Big]dr.
	\end{align*}
	It is easy to check that for each $\psi\in C_b(\bar{\mathscr{D}}), A\psi: \bar{\mathscr{D}}\to (0,\infty)$ is bounded continuous. Now
	\begin{eqnarray*}
		\lefteqn{ \|A\psi-A\tilde{\psi}\|}\\
		&=&\sup_{\bar{\mathscr{D}}}|A\psi-A\tilde{\psi}|\\
		& =& \sup_{\bar{\mathscr{D}}}\Big|\ds\sum_{l=0}^{n}P_{t,x,y}(\ell(t)=l)[\int_{0}^{T-t}e^{H_\theta(t,t+r,x)}\ds\sum_{j\neq x^l}p^l_{x^lj}(y^l+r)\times\\
		&&[\psi(t+r,R^l_jx,R^l_0(y+r\mathbf{1}))-\tilde{\psi}(t+r,R^l_jx,R^l_0(y+r\mathbf{1}))]f_{\tau^l|l}(r|x,y)dr]\Big|\\
		& \leq & \ds\sum_{l=0}^{n}P_{t,x,y}(\ell(t)=l)\int_{0}^{T-t}e^{H_\theta(t,t+r,x)}\ds\sum_{j\neq x^l}p^l_{x^lj}(y^l+r)f_{\tau^l|l}(r|x,y)dr\|\psi-\tilde{\psi}\|\\
		& < & K_1\|\psi-\tilde{\psi}\|,
	\end{eqnarray*}
	where $K_1:=\ds\sum_{l=0}^{n}P_{t,x,y}(\ell(t)=l)\int_{0}^{T-t}f_{\tau^l|l}(r|x,y)dr$, since the row sum of conditional probability matrix is $1$ and $h_\theta <0$ by Proposition \ref{hContinuityRS}(i). Since $F^l(\bar{y})$ is strictly less than $1$, \eqref{Ftaullv} implies that $F_{\tau^l|l}(r|x,y)<1$, for all $r\ge0$. Hence $K_1<1$.
	Therefore, $A$ is a contraction. Thus a direct application of Banach fixed point theorem ensures the existence
	and uniqueness of the solution to \eqref{intrisk}.
	
	\noi (ii) We denote the unique solution by $\psi$. Next we show that $\psi\in $ dom$(D_{t,y})$. To this end, it is sufficient to show that $A:C_b(\mathscr{D})\to \text{dom}(D_{t,y})$.
	The first term of $A\psi$ is in dom$(D_{t,y})$, which follows from Lemma \ref{theo1} (iv) and Proposition \ref{hContinuityRS} (ii). Now to show that the remaining term
	\begin{eqnarray*}
		\beta_l(t,x,y):=\int_{0}^{T-t}e^{H_\theta(t,t+r,x)}\ds\sum_{j\neq x^l}p^l_{x^lj}(y^l+r)\psi(t+r,R^l_jx,R^l_0(y+r\mathbf{1}))f_{\tau^l|l}(r|x,y)]dr,
	\end{eqnarray*}
	is also in the dom($D_{t,y}$) for any $\psi \in C(\mathscr{D})$, we need to check if the following limit
	\begin{align*}
	\lim_{\vp\rightarrow 0}\frac{1}{\vp} &\Big[\int_{0}^{T-t-\vp}e^{H_\theta(t+\vp,t+r+\vp,x)}\ds\sum_{j\neq x^l}p^l_{x^lj}(y^l+r+\vp)
	\psi(t+r+\vp,R^l_jx,R^l_0(y+(r+\vp) \mathbf{1}))\\
	&\times f_{\tau^l|l}(r|x,y+\vp))dr-\int_{0}^{T-t}e^{H_\theta(t,t+r,x)}\ds\sum_{j\neq x^l}p^l_{x^lj}(y^l+r)\psi(t+r,R^l_jx,R^l_0(y+r\mathbf{1}))\\
	&\times f_{\tau^l|l}(r|x,y)dr\Big],
	\end{align*}
	exists and, the limit is continuous in $\mathscr{D}$. If the limit exists, the limiting value is clearly
	$D_{t,y}\beta_l(t,x,y)$. By a suitable substitution of variables in the integral, the expression in the above limit can be rewritten, using \eqref{exftaul}, as
	\begin{eqnarray}\label{qeD}
	\no \lefteqn{\frac{1}{\vp}\Big[\int_{\vp}^{T-t}e^{H_\theta(t+\vp,t+r,x)}\ds\sum_{j\neq x^l}p^l_{x^lj}(y^l+r)\psi(t+r,R^l_jx,R^l_0(y+r\mathbf{1}))f_{\tau^l|l}(r-\vp|x,y+\vp)dr}\\
	\no &&-\int_{0}^{T-t}e^{H_\theta(t,t+r,x)}\ds\sum_{j\neq x^l}p^l_{x^lj}(y^l+r)\psi(t+r,R^l_jx,R^l_0(y+r\mathbf{1}))f_{\tau^l|l}(r|x,y)dr\Big]\\
	\no &=&\int_{0}^{T-t}e^{H_\theta(t,t+r,x)}\ds\sum_{j\neq x^l}p^l_{x^lj}(y^l+r) \psi(t+r,R^l_jx,R^l_0(y+r\mathbf{1}))\\
	\no && \times\frac{1}{\vp}\Big(e^{-H_\theta(t,t+\vp,x)} f_{\tau^l|l}(r-\vp|x,y+\vp) -f_{\tau^l|l}(r|x,y)\Big)dr -\frac{1}{\vp}\int_{0}^{\vp}e^{H_\theta(t+\vp,t+r,x)} \times \\
	&& \ds\sum_{j x^l}p^l_{x^lj}(y^l+r) \psi(t+r,R^l_jx,R^l_0(y+r\mathbf{1})) f_{\tau^l|l}(r-\vp|x,y+\vp) dr.
	\end{eqnarray}
	By Lemma \ref{theo1} (iv), $f_{\tau^l|l}(T-t|x,y)$ is in $dom(D_{t,y})$. Thus $D_{t,y}f_{\tau^l|l}(T-t|x,y)$ is bounded on $[0,T-t] $ by a positive constant $K_2$. Hence  by the mean value theorem on $f_{\tau^l|l}(T-t|x,y)$, the integrand of the first integral of \eqref{qeD} is uniformly bounded. Therefore, using the bounded convergence theorem, the integral converges as $\vp\rightarrow 0$.  The second integral of \eqref{qeD} converges as the integrand is continuous at $r=0$. Now we compute
	\begin{eqnarray*}
		\lefteqn{D_{t,y}\beta_l(t,x,y)}\\
		&=&\int_{0}^{T-t}e^{H_\theta(t,t+r,x)}\ds\sum_{j\neq x^l}p^l_{x^lj}(y^l+r) \psi(t+r,R^l_jx,R^l_0(y+r\mathbf{1}))\\
		\no && \Big(\frac{d}{dw} e^{-H_\theta(t,t+w,x)}\big|_{w=0}f_{\tau^l|l}(r|x,y)+f_{\tau^l|l}(r|x,y) \times \\
		\no && \lim_{\vp\to 0} \frac{1}{\vp} \left[ \frac{\int_0^\infty\prod_{m\neq l}(1-F^m(s+y^m|x^m))f^l(s+y^l|x^l)ds}
		{\int_0^\infty\prod_{m\neq l} (1-F^m(s+y^m+\vp|x^m)) f^l(s+y^l+\vp|x^l)ds} -1\right]\Big)dr\\
		&&-\ds\sum_{j\neq x^l}p^l_{x^lj}(y^l)\psi(t,R^l_jx,R^l_0y)f_{\tau^l|l}(0|x,y),
	\end{eqnarray*}
	using Lemma \ref{theo1} (iii). From \eqref{eq2} we know $\frac{\pa}{\pa y}\int_0^\infty\prod_{m\neq l}(1-F^m(s+y^m|x^m))f^l(s+y^l|x^l)ds = -\prod_{m\neq l}(1-F^m(y^m|x^m))f^l(y^l|x^l)$, therefore $D_{t,y}\beta_l(t,x,y)$ can be rewritten using \eqref{exftaul} as
	\begin{align}\label{Dbeta}
	\no& \int_{0}^{T-t}e^{H_\theta(t,t+r,x)}\ds\sum_{j\neq x^l}p^l_{x^lj}(y^l+r) \psi(t+r,R^l_jx,R^l_0(y+r\mathbf{1}))(- h_\theta(t,x)+f_{\tau^l|l}(0|x,y)) \times\\
	\no & f_{\tau^l|l}(r|x,y) dr-\ds\sum_{j\neq x^l}p^l_{x^lj}(y^l)\psi(t,R^l_jx,R^l_0y)f_{\tau^l|l}(0|x,y)\\
	=&[- h_\theta(t,x)+f_{\tau^l|l}(0|x,y)]\beta_l(t,x,y)-\ds\sum_{j\neq x^l}p^l_{x^lj}(y^l)\psi(t,R^l_jx,R^l_0y)f_{\tau^l|l}(0|x,y).
	\end{align}
	Clearly, \eqref{Dbeta} is in $C(\mathscr{D})$. Hence $\beta_l(t,x,y)$ is in the dom($D_{t,y}$). Hence the right hand side of \eqref{intrisk}  is in the dom($D_{t,y}$) for any $\psi\in C_b(\bar{\mathscr{D}}) $. Thus (ii) holds.
	\qed
	
	\subsection{The linear first order equation}\label{lfoe}
	\begin{Proposition}\label{IVPUniqRS}
		The unique solution to \eqref{intrisk} also solves the terminal value problem \eqref{eq:hjbPdeRS}-\eqref{eq:hjbPdeTerRS}.
	\end{Proposition}
	\proof
	Let $\psi$ be the solutions of the integral equation \eqref{intrisk}. Then by substituting $t=T$ in \eqref{intrisk}, \eqref{eq:hjbPdeTerRS} follows. Using the results from the proof of Lemma \ref{theo1}, Proposition \ref{IESolnRS}, Lemma \ref{theo1}(iv) and \eqref{Dbeta}, we have
	\begin{align*}
	D_{t,y}\psi(t,x,y)=&\ds\sum_{l=0}^{n}\Big[\ds\sum_{m=0}^n f_{\tau^m}(0|x^m,y^m)P_{t,x,y}(\ell(t)=l)-f_{\tau^l}(0|x^l,y^l)\Big][1-F_{\tau^l|l}(T-t|x,y)]\times\\
	&e^{H_\theta(t,T,x)}-\ds\sum_{l=0}^{n}P_{t,x,y}(\ell(t)=l)\Big[f_{\tau^l|l}(0|x,y)(F_{\tau^l|l}(T-t|x,y)-1)\Big]\times\\
	& e^{H_\theta(t,T,x)}-h_\theta(t,x)\ds\sum_{l=0}^{n}P_{t,x,y}(\ell(t)=l)[1-F_{\tau^l|l}(T-t|x,y)]\times\\
	&e^{H_\theta(t,T,x)}+\ds\sum_{l=0}^{n}\Big[\ds\sum_r f_{\tau^m}(0|x^m,y^m)P_{t,x,y}(\ell(t)=l)-f_{\tau^l}(0|x^l,y^l)\Big]\beta_l(t,x,y)\\
	&+\ds\sum_{l=0}^{n}P_{t,x,y}(\ell(t)=l)\Bigg(\Big(-h_\theta(t,x)+f_{\tau^l|l}(0|x,y)\Big)\beta_l(t,x,y)\\
	&-\ds\sum_{j\neq x^l}p^l_{x^lj}(y^l)\psi(t,R^l_jx,R^l_0y)f_{\tau^l|l}(0|x,y)\Bigg).
	\end{align*}
	Using \eqref{intrisk} and the equality in Lemma \ref{theo1}(v), the right hand side of above equation can be rewritten as
	\begin{align*}
	\ds\sum_{l=0}^n&\frac{f^l(y^l|x^l)}{1-F^l(y^l|x^l)}\Big[\psi(t,x,y)-\ds\sum_{j\neq x^l}p^l_{x^lj}(y^l)\psi(t,R^l_jx,R^l_0y)\Big]-h_\theta(t,x)\psi(t,x,y)\\
	&=-\ds\sum_{l=0}^n\ds\sum_{j\neq x^l}\la^l_{x^lj}(y^l)\Big[\psi(t,R^l_jx,R^l_0y)-\psi(t,x,y)\Big]-h_\theta(t,x)\psi(t,x,y).
	\end{align*}
	Hence $\psi$ satisfies \eqref{eq:hjbPdeRS}.
	\qed
	\begin{Proposition}\label{ClassicalRS}
		Let $\psi$ be a bounded classical solution to \eqref{eq:hjbPdeRS}-\eqref{eq:hjbPdeTerRS}. Then $\psi$ solves the integral equation \eqref{intrisk}.
	\end{Proposition}
	\proof If the PDE \eqref{eq:hjbPdeRS} has a classical solution $\psi$,  then $\psi$ is also in the domain of $\mathcal{A}$, where $\mathcal{A}$ is the infinitesimal generator of the Markov family $\{(r,X_r,Y_r)\}_{r\ge 0}$ starting from $(0,x_0,y_0)$ (say).
	Then we have from \eqref{eq:hjbPdeRS}
	\begin{eqnarray}\label{gen}
	\mathcal{A}\psi(t,x,y)+h_\theta(t,x)\psi(t,x,y)=0.
	\end{eqnarray}
	Consider
	$$\tilde{M}_t := e^{\int_0^t h_\theta(s,X_s)ds}\psi(t,X_t,Y_t).$$
	Then by It\^o's formula,
	$$d\tilde{M}_t=h_\theta(t,X_t)e^{\int_0^t h_\theta(s,X_s)ds}\psi(t,X_t,Y_t)dt+e^{\int_0^t h_\theta(s,X_s)ds}(\mathcal{A}\psi dt+dM_t^{(1)}),$$
	where $\{M_t^{(1)}\}_{t\ge 0}$ is a local martingale
	with respect to $\{\mathcal{F}_t\}_{t\ge 0}$, the usual  filtration generated by $\{(X_t,Y_t)\}_{t\ge 0}$. Thus from \eqref{gen} $\{\tilde{M}_t\}_{t\ge 0}$ is a local martingale.
	From definition of $\tilde{M}_t$, $\sup_{[0,T]}\tilde{M}_t<\|\psi\|e^{\|h_\theta\|T}$ a.s. Thus $\{\tilde{M}_t\}_{t\ge 0}$ is a martingale.
	Therefore by using \eqref{eq:hjbPdeTerRS}, we obtain
	\begin{eqnarray*}
		\psi(t,X_t,Y_t)= e^{\int_0^t -h_\theta(s,X_s)ds}\tilde{M}_t=\mathbb{E}[e^{\int_t^T h_\theta(s,X_s)ds}|\mathcal{F}_t]= [e^{\int_t^T h_\theta(s,X_s)ds}|X_t, Y_t]
	\end{eqnarray*}
	using the Markov property of $(X,Y)$. Thus
	\begin{equation}\label{eq:psi}
	\psi(t,x,y)= \mathbb{E}_{t,x,y}[e^{\int_{t}^{T}h_\theta(s,X_s)ds}],~\forall (t,x,y)\in\bar{\mathscr{D}}.
	\end{equation}
	By conditioning on the component of $X_t$ where the transition happens,
	\begin{eqnarray}\label{iec}
	\no \psi(t,x,y) &=&\mathbb{E}_{t,x,y}[\mathbb{E}_{t,x,y}[e^{[\int_{t}^{T}h_\theta(s,X_s)ds]}|\ell(t)]]\\
	&=&\ds\sum_{l=0}^{n}P_{t,x,y}(\ell(t)=l)\mathbb{E}_{t,x,y}[e^{[\int_{t}^{T}h_\theta(s,X_s)ds]}|\ell(t)=l]
	\end{eqnarray}
	where $\ell(t)$ is described in subsection 2.4 below (A4). Next by conditioning on $\tau^l(t)$ we rewrite
	\begin{eqnarray*}
		\no && \mathbb{E}_{t,x,y}[e^{\int_{t}^{T}h_\theta(s,X_s)ds}|\ell(t)=l]\\
		\no &=&\mathbb{E}_{t,x,y}[\mathbb{E}_{t,x,y}[e^{\int_{t}^{T}h_\theta(s,X_s)ds}|\ell(t)=l,\tau^l(t)]|\ell(t)=l]\\
		\no &=&P_{t,x,y}(\tau^l(t)>T-t|\ell(t)=l)e^{\int_{t}^{T} h_\theta(s,x)ds}\\
		&&+\int_{0}^{T-t}\mathbb{E}_{t,x,y}[e^{\int_{t}^{T}h_\theta(s,X_s)ds}|\ell(t)=l,\tau^l(t)=r] f_{\tau^l|l}(r|x,y) dr.
	\end{eqnarray*}
	Since $X$ is constant on $[t,t+r)$ provided $\ell(t)=l,\tau^l(t)=r$, the above expression is equal to
	\begin{eqnarray*}
		&&[1-F_{\tau^l|l}(T-t|x,y)]e^{H_\theta(t,T,x)}+\int_{0}^{T-t}e^{H_\theta(t,t+r,x)}\\
		&&\times \mathbb{E}_{t,x,y}[\mathbb{E}_{t,x,y}[e^{\int_{t+r}^{T}h_\theta(s,X_s)ds}|X^l_{t+r},\ell(t)=l,\tau^l=r]|\ell(t)=l,\tau^l=r]f_{\tau^l|l}(r|x,y)dr\\
		&=&[1-F_{\tau^l|l}(T-t|x,y)]e^{H_\theta(t,T,x)}\\
		&&+\int_{0}^{T-t}e^{H_\theta(t,t+r,x)}\times \ds\sum_{j\neq x^l}p^l_{x^lj}(y^l+r)\psi(t+r,R^l_jx,R^l_0(y+r\mathbf{1}))f_{\tau^l|l}(r|x,y)dr.
	\end{eqnarray*}
	From \eqref{iec} and the above expression, the desired result follows. \qed
	
	\proof[Proof of Theorem \ref{theo3}]
	The result follows from Proposition \ref{IESolnRS}, Proposition \ref{IVPUniqRS}, and Proposition \ref{ClassicalRS}.\qed
	
	\subsection{Optimal portfolio and verification theorem}
	
	Now we are in a position to derive the expression of optimal portfolio value under risk sensitive criterion. The optimal value is given by
	\begin{align}\label{IEFinalRS}
	\tilde{\varphi}_\theta(v,x,y) &:= \sup_u J_\theta^{u,T}(v,x,y) \nonumber \\
	& = -\frac{2}{\theta}\ln(\varphi_\theta(0,x,y,v)) \nonumber \\
	& = \ln(v)- \frac{2}{\theta}\ln(\psi(0,x,y)),
	\end{align}
	where the function $\varphi_\theta$ is defined in \eqref{expressvf} and $\psi$ is the unique classical solution to \eqref{eq:hjbPdeRS} - \eqref{eq:hjbPdeTerRS} obtained in Theorem \ref{theo3}.
	
	\begin{Remark}
		We note that the study of \eqref{eq:hjbPdeRS}-\eqref{eq:hjbPdeTerRS} becomes much simpler if the coefficients $r,\mu, \sigma$ are independent of time $t$. For time homogeneous case, Proposition \emph{\ref{hContinuityRS}} is immediate. Furthermore, the proof of Theorem \emph{\ref{theo3}} does not need the results given in Proposition \emph{\ref{IESolnRS}}, Proposition \emph{\ref{IVPUniqRS}}, and Proposition \emph{\ref{ClassicalRS}}. Indeed Theorem \emph{\ref{theo3}} can directly be proved by noting the smoothness of terminal condition.
	\end{Remark}
	We conclude this section with a proof of the verification theorem for optimal control problem \eqref{expressvf}. The main result is given in Theorem \ref{markov&admissible}.
	
	\begin{Proposition}\label{verificationth}
		Let $\vf_M$ be as in Theorem \emph{\ref{theo2}}, then
		\begin{itemize}
			\item[(i)] $\vf_M(t,x,y,v)\leq \tilde{J}^{\bar{u},T}_\theta(t,x,y,v)$ for every admissible Markov feedback control $\bar{u}$.
			\item[(ii)] Let $\bar{u}^\star$ be as in Proposition \emph{\ref{hContinuityRS}}(iv), then
            $$\vf_M(t,x,y,v)=\tilde{J}^{\bar{u}^\star,T}_\theta(t,x,y,v).$$
            Hence $\bar{u}^\star$ is optimal in the class of Markov feedback control.
		\end{itemize}
	\end{Proposition}
	\proof (i) Consider an admissible Markov feedback control $\bar{u}:=\{\bar{u}_t\}_{t\geq 0}$, where $\bar{u}_t=\tilde{u}(t,X_t,Y_t,V_t)$ and $\vf_M$, the classical solution to \eqref{hjbRS}-\eqref{hjbConsRS} as in \eqref{eq:RSTransf}. Now by It\^o's formula
	\begin{align}\label{itolem}
	&\vf_M(s,X_s,Y_s,V_s^{\bar{u}})-
	\vf_M(t,X_t,Y_t,V_t^{\bar{u}})- \int_t^s\ [\mathscr{A}^{\tilde{u}}\vf_M(r, X_r,Y_r,V_r^{\bar{u}})] dr \nonumber\\
	&= \displaystyle\ds\sum_{j=1}^{m_1}\int_t^s \frac{\partial}{\partial v}\vf_M(r,X_r,Y_r,V^{\bar{u}}_r) V^{\bar{u}}_r[\tilde{u}(r,X_r,Y_r,V_r)^*\si(r,X_r)]_j dW^j_r\nonumber \\
	&+\displaystyle\ds\sum_{j=1}^{m_2}\int_t^s \int_\mathbb{R} 	\Bigg[\vf_M(r,X_r,Y_r,V^{\bar{u}}_{r-}(1+[\tilde{u}(r,X_{r-},Y_{r-},V_{r-})^*\eta(z)]_j))- 	 \vf_M(r,X_r,Y_r,V^{\bar{u}}_{r-})\Bigg]\nonumber\\
	&\tilde{N}_j(dr,dz_j) +\ds\sum_{l=0}^{n}\int_t^s \int_\mathbb{R}
	\Bigg[\vf_M(r,R^l_{X^l_{r-}
		+h^l(X^l_{r-},Y^l_{r-},z_0)}(X_{r-}),R^l_{Y^l_{r-}-g^l(X^l_{r-},Y^l_{r-},z_0)}(Y_{r-}),V^{\bar{u}}_{r-}) \nonumber\\
	&-\vf_M(r,X_{r-},Y_{r-},V^{\bar{u}}_{r-})\Bigg]
	\tilde{\wp}^l(dr,dz_0).
	\end{align}
	We would first show that the right hand side is an $\{\mathscr{F}_s\}_{s\ge 0}$ martingale. Since $\bar{u}$ is admissible, using definition \ref{defi1}(iii), it is sufficient to show, the following square integrability condition
	\begin{align*}
	\mathbb{E}\int_t^s\left[ V^{\bar{u}}_r \frac{\partial}{\partial v} \vf_M(r,X_r,Y_r,V^{\bar{u}}_r)\right]^2dr
	< \infty,
	\end{align*}
	to prove that the first term is a martingale. Again since $\vf_M(t,x,y,v) = v^{-\frac{\theta}{2}}\psi(t,x,y)$, $v\frac{\pa \vf_M}{\pa v}=-\frac{\theta}{2}\vf_M= -\frac{\theta}{2} v^{-\frac{\theta}{2}}\psi(t,x,y)$. Thus using the  boundedness of $\psi$ the above would follow if
	\begin{align}\label{bndv}
	\mathbb{E}\int_t^s\left[ V^{\bar{u}}_r\right]^{-\theta}dr<\infty
	\end{align}
	holds. Now we consider the second integral. Rewriting that term, we obtain
	\begin{align}\label{2ndterm}
	\displaystyle\ds\sum_{j=1}^{m_2}\int_t^s \int_\mathbb{R}\left(V^{\bar{u}}_{r-}\right)^{-\frac{\theta}{2}}\psi(r,X_r,Y_r)
	\left[(1+[\tilde{u}(r,X_{r-},Y_{r-},V_{r-})^*\eta(z)]_j)^{-\frac{\theta}{2}}-1\right]\tilde{N}_j(dr,dz_j).
	\end{align}
	We first observe that $(1+[\tilde{u}(r,X_{r-},Y_{r-},V_{r-})^*\eta(z)]_j)>\delta$, and this implies $$(1+[\tilde{u}(r,X_{r-},Y_{r-},V_{r-})^*\eta(z)]_j)^{-\frac{\theta}{2}}<\delta^{-\frac{\theta}{2}}.$$ Thus the integrand of \eqref{2ndterm} is a product of a bounded function and $\left(V^{\bar{u}}_{r-}\right)^{-\frac{\theta}{2}}$. Since $\nu_j$, the L\'{e}vy  measure of $\tilde{N}_j$ is a finite measure for each $j$, to show \eqref{2ndterm} is an  $\{\mathscr{F}_s\}_{s\ge 0}$ martingale, it is enough to verify \eqref{bndv}. Similarly the third integral can be rewritten as
	\begin{align}\label{3rdterm}
	\ds\sum_{l=0}^{n}\int_t^s \int_{\mathbb{R}}\left(V^{\bar{u}}_{r-}\right)^{-\frac{\theta}{2}}
	&\Bigg[\psi(r,R^l_{X^l_{r-}
		+h^l(X^l_{r-},Y^l_{r-},z_0)}(X_{r-}),R^l_{Y^l_{r-}-g^l(X^l_{r-},Y^l_{r-},z_0)}(Y_{r-})) \no\\
	&-\psi(r,X_{r-},Y_{r-})\Bigg]\tilde{\wp}^l(dr,dz_0).
	\end{align}
	In \eqref{3rdterm} the integrand is a product of a bounded function with compact support and $\left(V^{\bar{u}}_{r-}\right)^{-\frac{\theta}{2}}$. Since, the compensator of $\tilde{\wp}^l(dr,dz_0)$ is $drdz_0$, \eqref{3rdterm} is also an  $\{\mathscr{F}_s\}_{s\ge 0}$ martingale if \eqref{bndv} holds. Thus \eqref{bndv} is the sufficient condition for the right side of \eqref{itolem} to be a martingale. However \eqref{bndv} readily follows from the Proposition \ref{solstate}(ii) and an application of Tonelli's Theorem.
	
	Taking conditional expectation on both sides of \eqref{itolem} given $X_t=x,Y_t=y,V^{\bar{u}}_t=v$ and letting $s\uparrow T$, we obtain
	\begin{align}\label{verification}
	&\mathbb{E}\left[(V^{\bar{u}}_T)^{\frac{\theta}{2}}|X_t=x,Y_t=y,V^{\bar{u}}_t=v\right]-\vf_M(t,x,y,v)\nonumber\\
	&=\mathbb{E}\int_t^T\bigg[\mathscr{A}^{\tilde{u}}\vf_M(r,
	X_r,Y_r,V_r^{\bar{u}})\bigg|X_t=x,Y_t=y,V^{\bar{u}}_t=v\bigg]dr\geq 0.
	\end{align}
	The above non-negativity follows, since $\vf_M$ is the classical solution to \eqref{hjbRS}-\eqref{hjbConsRS} and $\bar u_r \in \mathbb{A}_1$ for all $r$. Hence \eqref{expressvf} and \eqref{verification} imply result (i).
	
	(ii) The right hand side of \eqref{verification} becomes zero by considering $\tilde{u}(t,x,y,v)=u^\star(t,x)$ and this completes the proof of (ii).   \qed

	\noi Finally we show in the following theorem that $\vf_M$ as in Theorem \ref{theo2} indeed gives the optimal performance under all admissible controls.
	
	\begin{Theorem}\label{markov&admissible}
		Let  $\vf_M$ be as in Theorem \emph{\ref{theo2}} and $\vf_{A}:=\inf\{\tilde{J}_\theta^{u,T}(t,x,y,v):u={u(t,\omega)}~ \emph{\text{admissible control}}\}$. Then $\vf_M(t,x,y,v)= \vf_A(t,x,y,v)$.
	\end{Theorem}
	\proof We first note that in the proof of Proposition \ref{verificationth}(i), we have only used the properties (ii) and (iii) of Definition \ref{defi1} of the Markov control. Since these two properties are true for a generic admissible control $u$, we can get as in Proposition \ref{verificationth}(i). $$\vf_M(t,x,y,v)\leq \tilde{J}^{u,T}_\theta(t,x,y,v)$$ for every admissible control $u$. By taking infimum, we get $\vf_M\leq \vf_A$. The other side of inequality is rather straight forward. Using Proposition \ref{verificationth}(ii) and Theorem \ref{hContinuityRS}(iv), $\bar{u}^\star$ is admissible, and $\vf_M(t,x,y,v)= \tilde{J}^{\bar{u}^\star,T}_\theta(t,x,y,v)$. Thus $\vf_M\geq\vf_A$. Hence the result is proved.\qed
	
	\noi Now we establish a characterisation of $\vf_M$ using the HJB equation in the following Proposition.
	
	\begin{Proposition}\label{comparison}
		Let $\vf$ be any classical solutions to \eqref{hjbRS}-\eqref{hjbConsRS}. Let  $\vf_M$ be as in Theorem \emph{\ref{theo2}}. Then $\vf_M(t,x,y,v)\geq \vf(t,x,y,v)$, for all $t,x,y,v$. Thus the unique solution $\vf_M$ obtained in Theorem \ref{theo2} is maximal among all classical solution to \eqref{hjbRS}-\eqref{hjbConsRS}.
	\end{Proposition}
	\proof Note that in the Proof of Proposition \ref{verificationth}(i), to show that the right hand side of \eqref{itolem} is a martingale, we have only effectively used the fact that $\vf_M$ satisfies  conditions (i),(ii) and (iii) of Definition \ref{defiG}. Hence for any $\vf \in \mathscr{G}$ and $\bar{u}^\star$ as in  Proposition \ref{hContinuityRS}(iv),
	\begin{align}\label{newitolem}
	\vf(s,X_s,Y_s,V_s^{\bar{u}^\star})-
	\vf(t,X_t,Y_t,V_t^{\bar{u}^\star})- \int_t^s\ [\mathscr{A}^{u^\star}\vf(r, X_r,Y_r,V_r^{\bar{u}^\star})] dr,
	\end{align}
	is an $\{\mathscr{F}_s\}_{s \ge0}$ martingale. Taking conditional expectation in \eqref{newitolem}, given  $X_t=x,Y_t=y,V^{\bar{u}^\star}_t=v$ and letting $s\uparrow T$, we have
	\begin{align*}
	&\mathbb{E}\left[(V^{\bar{u}^\star}_T)^{-\frac{\theta}{2}}|X_t=x,Y_t=y,V^{\bar{u}^\star}_t=v\right]-\vf(t,x,y,v)\nonumber\\
	&=\mathbb{E}\int_t^T\bigg[\mathscr{A}^{u^\star}\vf(r,
	X_r,Y_r,V_r^u)\bigg|X_t=x,Y_t=y,V^{\bar{u}^\star}_t=v\bigg]dr,
	\end{align*}
	using $\vf(T,X_T,Y_T,V^{\bar{u}^\star}_T)=\left(V^{\bar{u}^\star}_T\right)^{-\frac{\theta}{2}}$. Now using nonnegativity of right side and Proposition \ref{verificationth}(ii), we obtain $\vf_M(t,x,y,v)\geq \vf((t,x,y,v)$.\qed

	\section{Numerical Example} \label{sec:NumSim}
	We have seen that the optimal portfolio value with risk sensitive criterion is given by (\ref{IEFinalRS}) and \eqref{eq:hjbPdeRS} - \eqref{eq:hjbPdeTerRS}. For illustration purpose, we are considering a simple model in which all the parameters for all assets are governed by a single semi-Markov process. Then $h_\theta(t,x)= h_\theta(t,\bar{x})$ if $ \bar{x}^0=x^0$, and we denote that value as  $\bar{h}_\theta(t,x^0)$ where $x^0$ and $y^0$ are the first components of $x$, and $y$ respectively. Hence \eqref{eq:psi} implies $\psi(t,x,y) = \psi(t,\bar{x},\bar{y})$ provided $\bar{x}^0=x^0$ and $\bar{y}^0=y^0$. In other words $ \psi(t,x,y)$ depends only on $(t,x^0,y^0)$. In view of this, we may introduce a new function $\bar{\psi} (t,x^0,y^0)$ to denote $\psi (t,(x^0,\ldots, x^n),(y^0,\ldots, y^n))$. Therefore \eqref{eq:hjbPdeRS} gets reduced to
	\begin{align}\label{eq1}
	D_{t,y}\psi(t,x,y)+\ds\sum_{j\neq x}\la^0_{xj}(y)\Big[\psi(t,j,0)-\psi(t,x,y)\Big]+\bar{h}_\theta(t,x)\psi(t,x,y)=0,
	\end{align}
	for every $x\in \mathcal{X}$, $y\in (0,t)$, $t\in (0,T)$. We further assume that $n=1$, i.e., the portfolio includes a single stock and a money market instrument. We also specify the state space $\mathcal{X}=\{1,2,3\}$, i.e., the semi-Markov process has three regimes. The drift coefficient, volatility and instantaneous interest rate at each regime are chosen as follows:
	\[\left(\mu(i),\sigma(i),r(i)\right) = \left\{
	\begin{array}{lr}
	(0.3,0.2,0.2) & : i=1\\
	(0.6,0.4,0.5) & : i=2\\
	(0.8,0.3,0.7) & : i=3.
	\end{array}
	\right.
	\]
	The transition rates for $i\neq j$ are assumed to be given by
	$$\la^0_{ij}(y)= (y-\ln(1+y))p_{ij}
	$$
	where
	\[ (p_{ij})_{ij} = \left( \begin{array}{ccc}
	0 & 2/3 & 1/3 \\
	1/2 & 0 & 1/2 \\
	1/3 & 2/3 & 0 \end{array} \right).\]
	\noi Hence the holding time of the first component in each regime has the conditional probability density function $f(y \mid i) = y\exp(-y)$ and the conditional c.d.f $f(y \mid i)= 1-(1+y)e^{-y}$.
	We also assumed $\eta(z) = z$ and $\nu(dz) := \frac{1_{[a,b]}(z)}{b-a} dz$.
	
	It is shown separately in \cite{GGS09} that the classical solution to \eqref{eq1} with $\bar{\psi}_\theta(T,x,y)=1$, satisfies the following integral equation
	\begin{align}\label{eq:IERS}
	\bar{\psi}(t,x,y) & = \frac{1 - F(T-t+y \mid x)}{1 - F(y \mid x)}\exp \left[\Int_{t}^{T}\bar h_\theta(s,x)\,ds\right] + \Int_{0}^{T-t}\exp\left[\Int_{t}^{t+r}\bar h_\theta(s,x)\,ds\right]\times  \nonumber \\
	& \quad\ds\sum_{j \neq x} p_{xj}(y+r)\bar{\psi}(t+r,j,0)\frac{f(y+r \mid x)}{1-F(y \mid x)}dr,
	\end{align}
	which also follows from \eqref{intrisk}. Here we compute $\bar{\psi}_\theta(t,x,y)$ by discretization of above integral equation using an implicit step-by-step quadrature method as developed in \cite{GGS09}. We take $T =1$, $\Delta t= 0.002$ so $m = 0,1,2,\ldots, M=\lfloor\frac{T}{\Delta t}\rfloor$. The discretization is given by
	\begin{align*}
	\psi^m(i,y) & \approx \bar \psi(T-m\Delta t,i,y).
	\end{align*}
	Therefore from (\ref{eq:IERS}) we get
	\begin{align}\label{eq:psiTildeIter}
	\psi^m(i,y) =& \frac{1 - F(m\Delta t+y \mid i)}{1 - F(y \mid i)}\exp\left[H_\theta^0(i) - H_\theta^m(i)\right]  + \Delta t \ds\sum_{l=0}^{m}w_m(l)\nonumber \\
	& \frac{f(y+l\Delta t \mid i)}{1-F(y \mid i)} \left(\exp\left[H_\theta^{m-l}(i) - H_\theta^m(i)\right] \ds\sum_{j \in\mathcal{X}, j \ne i}p_{ij}\psi^{m-l}(j,0)\right),
	\end{align}
	where $w_m(l)$ are weights, chosen as below
	\begin{equation}
	w_m(l) = 1 ~\textrm{for} ~l=1, 2, \ldots, m-1, ~~ w_m(0) = w_m(m) = \frac{1}{2},\nonumber
	\end{equation}
	and
	\begin{align*}
	\nonumber H_\theta^m(i) & ~:=~ \Int_{0}^{T-m\Delta t} \bar h_\theta(s,i)\,ds, \\
	\bar h_\theta(t,i) & = \inf_{u \in \mathbb{A}_1}\left[-\frac{\theta}{2}\left[r(t,i) + b(t,i)~u\right]+ \frac{1}{2}  \left(-\frac{\theta}{2}\right) \left(-\frac{\theta}{2}-1\right)u^2\sigma^2(t,i) \right.\nonumber \\
	& \quad \quad \quad \left.  -1+ \frac{(1+bu)^{1-\frac{\theta}{2}} - (1+au)^{1-\frac{\theta}{2}}}{u(1-\frac{\theta}{2})(b-a)} \right].
	\end{align*}
	\noi For a given initial portfolio value $v$, from \eqref{IEFinalRS} and \eqref{eq:psiTildeIter} we get
	\begin{align}\label{eq:phiTildeIter}
	\tilde\vf_\theta(v,i,y) \approx \ln(v) -  \frac{2}{\theta}\ln(\psi^M(i,y)).
	\end{align}
	Thus the numerical approximation of risk sensitive optimal wealth is given by  \eqref{eq:psiTildeIter}-\eqref{eq:phiTildeIter}.
	
	In Proposition \ref{hContinuityRS} we have seen that there exists a unique $u \in \mathbb{A}_1$ which gives $h_\theta(t,i)$ and that we can find by using any convex optimization technique. Here we have used the interior-point method to find the optimal $u$.
	
	\begin{figure}[h]
		\centering
		\begin{subfigure}{.49\textwidth}
			\centering
			\includegraphics[width=7cm]{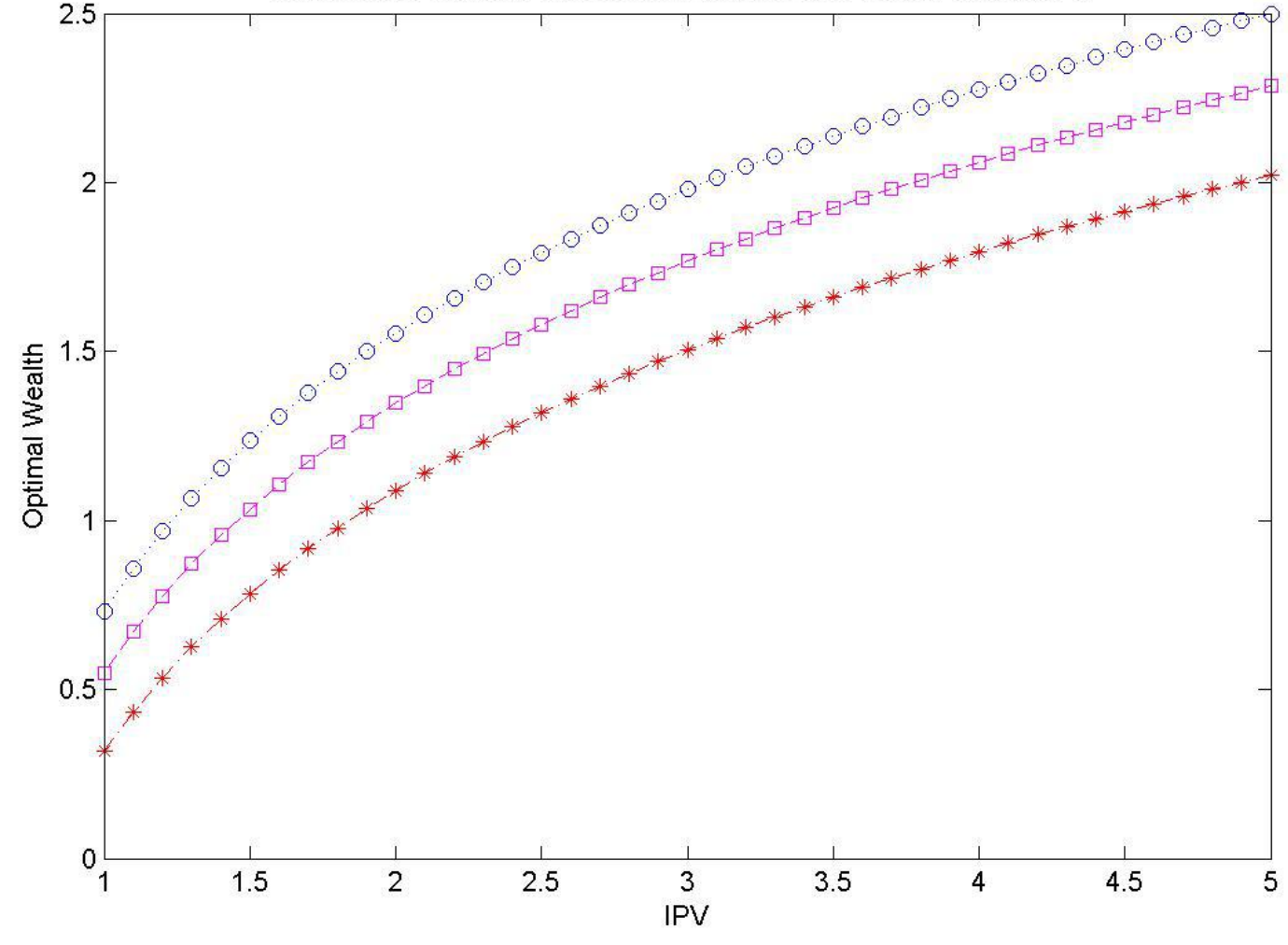}
		\end{subfigure}
		\begin{subfigure}{.49\textwidth}
			\centering
			\includegraphics[width=7cm]{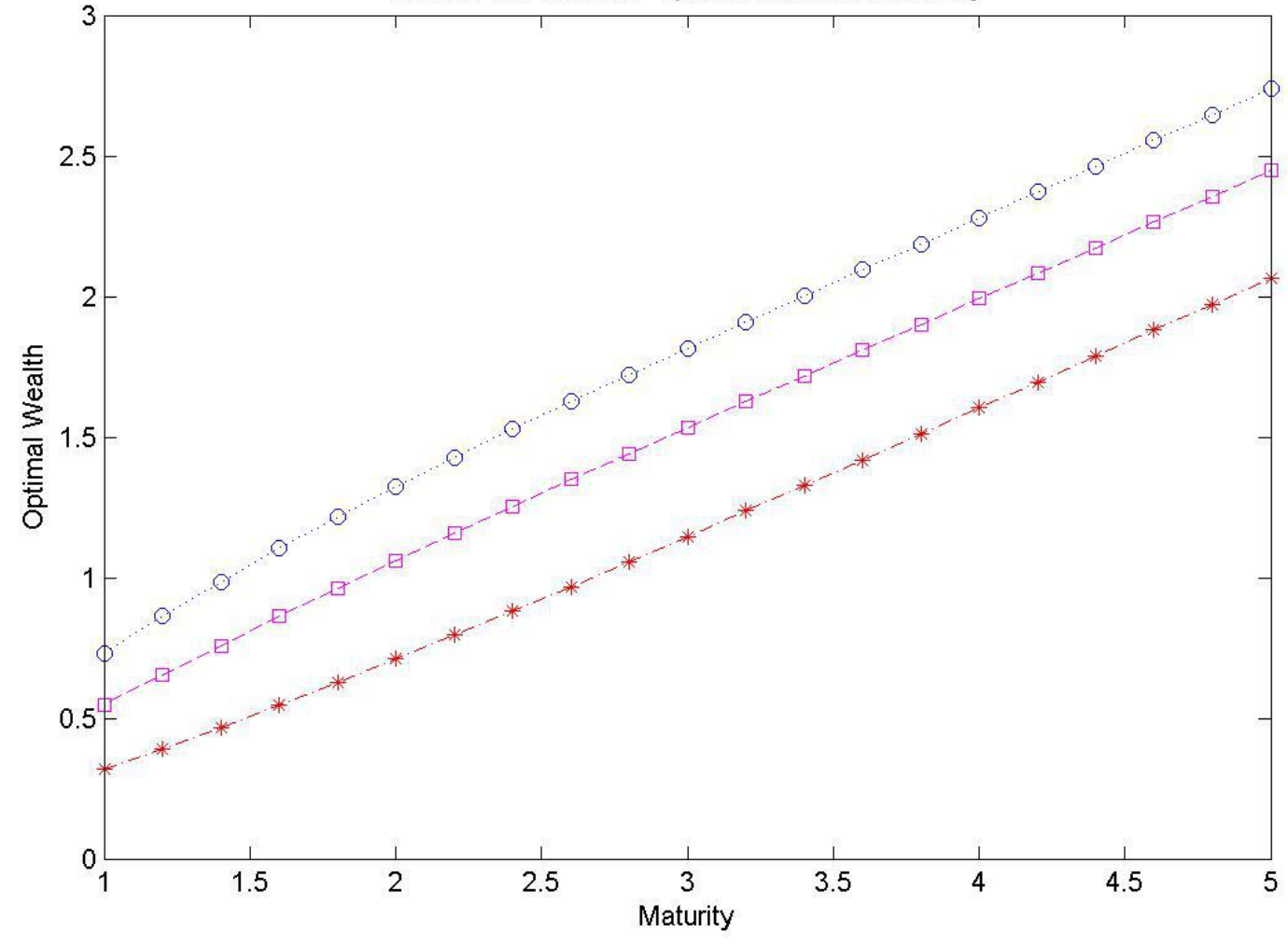}
		\end{subfigure}
		\caption{Finite horizon risk sensitive optimal wealth function}\label{fig1}
	\end{figure}
	
	\begin{figure}[h]
		\centering
		\includegraphics[width=10cm]{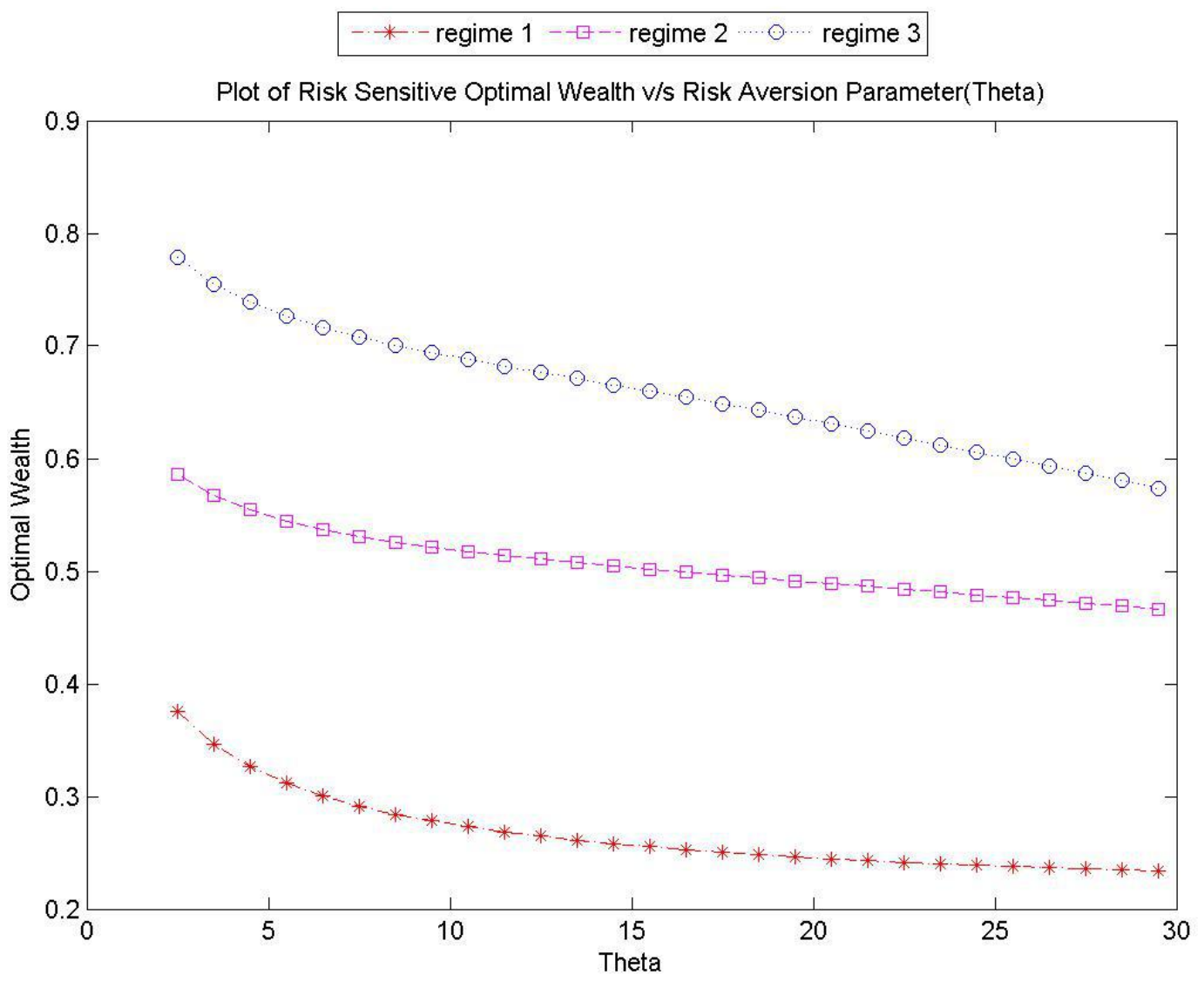}
		\caption{ Optimal wealth function vs risk aversion parameter}\label{fig:fig2RS}
	\end{figure}
	
	\noi We use above mentioned numerical scheme to compute the risk sensitive optimal wealth function given in (\ref{eq:phiTildeIter}). In all $3$ figures each line corresponds to a particular value of $i$. To be more precise the cross line corresponds to $i=1$, whereas the box and circle lines are for $i=2$ and $i=3$. Figure \ref{fig1} describes the behavior of risk sensitive optimal wealth for different values of initial portfolio wealth and maturity. The left side plot in figure \ref{fig1} shows that the optimal wealth is increasing and concave with the value of initial investment. This is due to the concavity of our objective function. On the other hand the right side plot shows linearity of the optimal wealth with respect to the the maturity of investment. We also note a strict hierarchy  of optimal wealth according to the market parameter values at different regimes. However a detailed analysis based on series of numerical experiments may reveal some finer sensitivity results. We refrain to discuss those in this paper.  Figure \ref{fig:fig2RS} shows the movement in risk sensitive optimal wealth for different values of risk aversion parameter. The plot shows the strict diminishing behavior of risk sensitive optimal wealth for increasing risk aversion parameter value. This observation is consistent to the common sense ``no risk, no gain''.	
	
	\section{Conclusion} \label{sec:Conc}
	In this paper a portfolio optimization problem, without any consumption and transaction cost, where stock prices are modelled by multi dimensional geometric jump diffusion market model with semi-Markov modulated coefficients is studied. We find the expression of optimal wealth for expected terminal utility method with risk sensitive criterion on finite time horizon. We have studied the existence of classical solution of HJB equation using a probabilistic approach. We have obtained the implicit expression of optimal portfolio. It is important to note that, the control is robust in the sense that the optimal control does not depend on the transition function of the regime. We have also implemented a numerical scheme to see the behavior of solutions with respect to initial portfolio value, maturity and risk of aversion parameter. The results of the numerical scheme are in agreement with the theory of financial market. The corresponding problem in infinite horizon is needs further investigation. This would  require appropriate results on large deviation principle for semi-Markov processes which need to be carried out.

	\appendix
	\section{Proof of Lemmata}
		\begin{Lemma}\label{expeta1}
			Let $\bar{N}$ be a Poisson random measure on $[0, \infty)\times \mathbb{R}$ defined on the probability space $(\Omega, \mathscr{F}, P)$ with intensity $dt\bar{\nu}(dz)$, where $\bar{\nu}$ is a finite measure. If $\bar{\eta}\in L^1(\nu)$, then there exists a positive constant $c$ such that
			\begin{align*}	\mathbb{E}\left[ \exp\left(\Int_{0}^{t  }\! \Int_{\mathbb{R}}\ln(1+ \bar\eta(z))\,\bar N(ds, dz)\right)\right]=\exp\left(ct\bar \nu(\RR)\right).
			\end{align*}
			
		\end{Lemma}
		\proof We first note that $|\bar{N}_t|:=\bar{N}([0,t]\times\mathbb{R})$ is finite a.s. as $|\bar{\nu}|<\infty$. Therefore the integral $\int_{0}^{t}\int_{\mathbb{R}}\ln(1+ \bar\eta(z))\,\bar N(ds, dz)$ can be written as $\sum_{i=1}^{|\bar{N}_t|}\ln(1+ \bar\eta(z_i))$, where $\{(t_i,z_i)\mid i=1,\ldots, |\bar{N}_t|\}$ are the point masses of $\bar{N}$ on $[0,t]\times\mathbb{R}$. To be more precise,  $\bar N(A)=\sum_{i=1}^{|\bar{N}_t|}\delta_{\{(t_i,z_i)\}}(A)$ for all $A \in\mathscr{B}([0,t]\times\mathbb{R})$. Therefore
		\begin{align}\label{inteta1}
		\no \mathbb{E}\left[ \exp\left(\Int_{0}^{t  }\! \Int_{\mathbb{R}}\ln(1+ \bar\eta(z))\,\bar N(ds, dz)\right)\right]&=\mathbb{E}\left[\prod_{i=1}^{|\bar{N}_t| }(1+ \bar\eta(z_i))\right]\\
		&= \mathbb{E}\left[\mathbb{E}\left[\prod_{i=1}^{|\bar{N}_t| }(1+ \bar\eta(z_i))\bigg||\bar{N}_t|\bigg|\right]\right].
		\end{align}
		Since $(1+\bar\eta(z_1)),\ldots,(1+\bar\eta(z_{|\bar N_t|}))$ are conditionally independent and identically distributed given $|\bar N_t|=n$, the right side is equal to
		\begin{align*}
		\ds\ds\sum_{n=1}^{\infty}\mathbb{E}[(1+\bar \eta(z_1))]^nP(|\bar N_t|=n).
		\end{align*}
		Now using $\mathbb{E}\left[\bar \eta(z_1)\right]=c<\infty$, and $P(|\bar{N}_t|=n)=e^{-t\bar \nu(\RR)}\frac{(t\bar\nu(\RR))^n}{n!}$, the above sum is equal to
		\begin{eqnarray*}
		\lefteqn{\sum_{n=1}^{\infty}(1+c)^ne^{-t\bar \nu(\RR)}\frac{(t\bar\nu(\RR))^n}{n!}}\\
		&=&e^{-t\bar \nu(\RR)}\exp\left(t\bar \nu(\RR)(1+c)\right)\\
		&=&\exp\left(ct\bar \nu(\RR)\right).
		\end{eqnarray*} Hence the proof. \qed
	\begin{proof}[Proof of Lemma \ref{SDESol}]
		First we show the uniqueness by assuming that the SDE
		(\ref{eq:sde}) admits  a solution, $\{S_t^l\}_{t\geq 0}$, say, the stopping time  $\tau =\min\{t \in [0,\infty) \mid S_t^l \leq 0\}$. Using It\^ o  Lemma (Theorem 1.16 of \cite{OkSu}) for $ 0\leq s < t \wedge \tau$ we get,
		\begin{align*}
		d\ln(S_s^l) & = \frac{S_{s-}^l}{S_{s-}^l}\left[\mu^l(s, X_{s-} )ds + \ds\sum_{j=1}^{m_1} \sigma_{lj}(s, X_{s-}) ~dW_s^j \right] -\frac{1}{2}(S_{s-}^l)^{-2}(S_{s-}^l)^2 a_{ll}(s,X_{s-})ds \\
		& \quad + \ds\sum_{j=1}^{m_2}\Int_{\mathbb{R}}\! \left[\ln(S_{s-}^l + \eta_{lj}(z_{j})S_{s-}^l)
		-\ln(S_{s-}^l)\,  \right]\, N_j(ds, dz_j).
		\end{align*}
		Integrating both sides from 0 to $t \wedge\tau$ yields,
		\begin{align*}
		\ln\left(S^l_{t\wedge\tau}\right)-\ln s_l & =\Int_0^{t\wedge\tau}\left(\mu^l(s, X_{s-}) - \frac{1}{2} a_{ll}(s,X_{s-}) ds\right)+ \ds\sum_{j=1}^{m_1} \Int_{0}^{t \wedge\tau}\! \sigma_{lj}(s, X_{s-})\,dW_s^j \nonumber \\
		& + \ds\sum_{j=1}^{m_2}\Int_{0}^{t \wedge\tau}\! \Int_{\mathbb{R}}\ln(1+ \eta_{lj}(z_j))\,N_j(ds, dz_j),
		\end{align*}
		where all the integrals have finite expectations almost surely by using \textbf{(A2)}.
		\begin{align}
		S_{t \wedge\tau}^l & = s_l\exp\Bigg[\Int_{0}^{t \wedge\tau} \!\left(\mu^l(s, X_{s-}) - \frac{1}{2} a_{ll}(s,X_{s-})\right)ds+ \ds\sum_{j=1}^{m_1} \Int_{0}^{t \wedge\tau}\! \sigma_{lj}(s, X_{s-})\,dW_s^j \nonumber\\
		& +\ds\sum_{j=1}^{m_2}\Int_{0}^{t \wedge\tau}\! \Int_{\mathbb{R}}\ln(1+ \eta_{lj}(z_j))\,N_j(ds, dz_j)\Bigg]
		\end{align}
		Thus any solution to (\ref{eq:sde}) has the above expression. Under \textbf{(A2)}, $\Int_{0}^{\tau} \Int_{\mathbb{R}}\ln(1+ \eta_{lj}(z_j))\,N_j(ds, dz_j)$ has finite expectation for any finite stopping time $\tau$.
		
		Let  $\Omega_1:=\{\omega \in \Omega: \tau(\omega)<\infty\}$. Now
		if possible, assume $P(\Omega_1)>0$. By letting $t \rightarrow
		\infty $ in the above expression, we obtain that
		$S_{\tau(\omega)-}^l$ is exponential of a random variable which is
		finite for almost every $\omega \in \Omega_1$. Thus
		$S_{\tau(\omega)-}^l > 0$. But for almost every $\omega\in
		\Omega_1$ $S_{\tau(\omega)}^l\le 0$. Hence non-positivity occurred
		only by jump. In other words $\eta_{lj}(z_j)\le -1$ for some
		$z_j$. But that is contrary to the assumption on $\eta$. Hence
		$\tau = \infty ~ P~a.s$. Therefore, $S_t^l > 0~
		P$ a.s. for all $t \in (0,\infty)$ and is given by
		\begin{align}\label{eq:sdeSol}
		S_{t}^l & = S_0^l\exp\left[\ds\sum_{j=1}^{m_1} \Int_{0}^{t  }\! \sigma_{lj}(s, X_{s-})\,dW_s^j+
		\ds\sum_{j=1}^{m_2}\Int_{0}^{t  }\! \Int_{\mathbb{R}}\ln(1+ \eta_{lj}(z_j))\,\bar{N}_j(ds, dz_j)   \right.\nonumber \\
		& \quad  \left.+\Int_{0}^{t  } \!\Bigg(\mu^l(s, X_{s-}) - \frac{1}{2} (\sigma_l(s,X_{s-}) \sigma_l(s,X_{s-})^{*})\right. \nonumber \\
		& \quad \left.  + \ds\sum_{j=1}^{m_2}\! \Int_{|z_j|<1}\!\left(\ln(1+ \eta_{lj}(z_j)) -\eta_{lj}(z_j)\right) \,\nu_j(dz_j)\Bigg)\,ds  \right].
		\end{align}
		Thus by equation (\ref{eq:sdeSol}), $S^l=\{S_t^l\}_{t \geq 0}$ is an adapted and rcll process and is
		uniquely determined with the initial condition $S_0^l = s_0$. Hence the solution is unique.
		
		\noi It is easy to show by a direct calculation that the process
		$S^l$, given by \eqref{eq:sdeSol} indeed solves the SDE
		\eqref{eq:sde}.
	\end{proof}
	\proof[Proof of Lemma \ref{theo1}]
	\noi (i) One can compute the conditional c.d.f
	$F_{\tau^l}(\cdot|i,\bar{y})$ in the following way
	\begin{eqnarray}\label{eqf}
	\no F_{\tau^l}(s|i,\bar{y})&=&P(0\leq\tau^l(t)\leq s|X^l_t=i,Y^l_t=\bar{y})\\
	\no &=& P(\tau^l(t)+Y^l_t\leq s+\bar{y}|X^l_t=i,Y^l_t=\bar{y})\\
	\no &=& P(Y^l_{T^l_{n^l(t)+1}-}\leq s+\bar{y}|Y^l_{T^l_{n^l(t)}-}\geq \bar{y},X^l_t=i,Y^l_t=\bar{y})\\
	& = & \frac{F^l(s+\bar{y}|i)-F^l(\bar{y}|i)}{1-F^l(\bar{y}|i)}\quad\quad l=0,1,\ldots,n .
	\end{eqnarray}
	We also denote the derivative of $F_{\tau^l}(s|i,\bar{y})$ by $f_{\tau^l}(s|i,\bar{y})$, given by
	\begin{eqnarray}\label{eqft}
	f_{\tau^l}(\cdot|i,\bar{y})=\frac{f^l(\cdot+\bar{y}|i)}{1-F^l(\bar{y}|i)}.
	\end{eqnarray}
	From the definition of $F_{\tau^l|l}(v|x,y)$ we have,
	\begin{eqnarray}\label{eq:F}
	\no F_{\tau^l|l}(v|x,y)&=&P_{t,x,y}(\tau^l(t)\leq v|\ell(t)=l)\\
	& = & \frac{P_{t,x,y}(\tau^l(t)\leq v,\ell(t)=l)}{P_{t,x,y}(\ell(t)=l)}.
	\end{eqnarray}
	We also introduce a new variable $\tau^{-l}(t):=\min_{m\neq l}\tau^m(t)$.
	We denote the conditional c.d.f of $\tau^{-l}(t)$ given $X_t=x$ and $Y_t=y$
	as $F_{\tau^{-l}}(\cdot|x,y)$ which is equal to $1-\prod_{m\neq l} (1-F_{\tau^m}(\cdot|x^m,y^m))$.
	
	It is easy to see that $P_{t,x,y}(\tau^l(t)\leq v,\ell(t)=l)=P_{t,x,y}(\tau^{-l}(t)>\tau^l(t),\tau^l(t)\leq v)$.
	To compute this probability we use a conditioning on $\tau^l(t)$. Thus
	\begin{eqnarray}\label{ptxytault}
	\no  P_{t,x,y}(\tau^l(t)\leq v,\ell(t)=l)&=& \mathbb{E}_{t,x,y}[P_{t,x,y}(\tau^{-l}(t)>\tau^l(t),\tau^l(t)\leq v|\tau^l(t))]\\
	\no & = & \Int_0^v P_{t,x,y}(\tau^{-l}(t)>\tau^l(t)|\tau^l(t)=s) f_{\tau^l}(s|x^l,y^l)ds\\
	\no & = & \Int_0^v(1-P_{t,x,y}(\tau^{-l}(t)\leq s))f_{\tau^l}(s|x^l,y^l)ds\\
	& = & \Int_0^v\prod_{m\neq l}(1-F_{\tau^m}(s|x^m,y^m))f_{\tau^l}(s|x^l,y^l)ds.
	\end{eqnarray}
	Again, $P_{t,x,y}(\ell(t)=l)=P_{t,x,y}(\tau^l(t)\leq \infty,\ell(t)=l)$ and from \eqref{eqf}, \eqref{eqft} we have (i).

	\noi (ii) From \eqref{eq:F}, one gets \eqref{Ftaullv}. Since $\lambda^l$ is $C^1$ in $s$, $\prod_{m\neq l}(1-F^m(s+y^m|x^m))f^l(s+y^l|x^l)$ is $C^1$ in $s\in [0,T]$. Thus
	by fundamental theorem of calculus, $F_{\tau^l|l}(v|x,y)$ is twice differentiable wrt $v$.
	
	\noi (iii) Follows directly from (ii).
	
	\noi (iv) In order to show that $P_{t,x,y}(\ell(t)=l)$ and $F_{\tau^l|l}(T-t|x,y)$ belong to $D_{t,y}$ we introduce a new function
	$\digamma^l_v(x,y):=\int_0^v\prod_{m\neq l}(1-F^m(s+y^m|x^m))f^l(s+y^l|x^l)ds$ and $\digamma^l_\infty(x,y):=\displaystyle\lim_{v\rightarrow\infty}\digamma^l_v
	(x,y)$. Consider another function
	\begin{equation}\label{digprim}
	\digamma^{l'}_v(x,y):=\prod_{m\neq l}(1-F^m(v+y^m|x^m))f^l(v+y^l|x^l).
	\end{equation}
	We note that $\digamma^{l'}_v(x,y)$ is the  derivative of $\digamma^l_v(x,y)$ with respect to $v$ and it is continuous. Now we show that $\digamma^l_v(x,y)$ is $C^1$ in $y$. To this end we first show the existence of the following limit
	\begin{eqnarray*}
		\lim_{\vp\to 0}&\frac{1}{\vp}\Big  [\int_{0}^{v}\prod_{m\neq l}(1-F^m(s+y^m+\vp|x^m))f^l(s+y^l+\vp|x^l)ds\\
		&-\int_{0}^{v}\prod_{m\neq l}(1-F^m(s+y^m|x^m))f^l(s+y^l|x^l)ds\Big].
	\end{eqnarray*}
	By a suitable substitution of variable, the expression in the above limit is
	\begin{eqnarray*}
		&\frac{1}{\vp}\Big  [\int_{v}^{v+\vp}\prod_{m\neq l}(1-F^m(s+y^m|x^m))f^l(s+y^l+\vp|x^l)ds\\
		&-\int_{0}^{\vp}\prod_{m\neq l}(1-F^m(s+y^m|x^m))f^l(s+y^l|x^l)ds\Big].
	\end{eqnarray*}
	Using \eqref{digprim} the above expression converges to $\digamma^{l'}_v(x,y)-\digamma^{l'}_0(x,y)$ as $\vp\rightarrow 0$ and the limit is continuous in $y$. Thus $$D_{t,y}\digamma^l_v(x,y)=\digamma^{l'}_v(x,y)-\digamma^{l'}_0(x,y).$$ If $v$ is a differentiable function of $t$, then $$D_{t,y}\digamma^l_v(x,y)=\digamma^{l'}_v(x,y)\left(1+\frac{\partial v}{\partial t}\right)-\digamma^{l'}_0(x,y).$$
	Hence
	\begin{equation}\label{eq2}
	D_{t,y}\digamma_v^l(x,y)=
	\begin{cases}
	\digamma^{l'}_v(x,y)\left(1+\frac{\partial v}{\partial t}\right)-\digamma^{l'}_0(x,y) & 0<v<\infty \\
	-\digamma^{l'}_0(x,y) & v=\infty.
	\end{cases}
	\end{equation}
	Since $$D_{t,y}\prod_m(1-F^m(v+y^m|x^m))=-\ds\sum_rf^r(y^r|x^r)\prod_{m\neq r}(1-F^m(y^m|x^m)$$
	it follows from Lemma \ref{theo1} (i), (ii) and the above notations  $P_{t,x,y}(\ell(t)=l)=\frac{\digamma^l_\infty(x,y)}{\prod_m(1-F^m(y^m|x^m))}$ and $F_{\tau^l|l}(T-t|x,y)=\frac{\digamma^l_{T-t}(x,y)}{\digamma^l_\infty(x,y)}$. Hence $P_{t,x,y}(\ell(t)=l)$ and $F_{\tau^l|l}(T-t|x,y)$ are in the dom($D_{t,y}$). Now operating $D_{t,y}$ on $P_{t,x,y}(\ell(t)=l)$  and using \eqref{eqft}, \eqref{digprim} we have
	\begin{eqnarray*}
		D_{t,y}P_{t,x,y}(\ell(t)=l)&=&\frac{D_{t,y}\digamma_\infty(x,y)}{\prod_m(1-F^m(v+y^m|x^m))}\\
		&&+\frac{\digamma_\infty(x,y)\times\ds\sum_rf^r(y^r|x^r)\prod_{m\neq r}(1-F^m(y^m|x^m))}{(\prod_m(1-F^m(v+y^m|x^m)))^2}\\
		&=&-\frac{\digamma'_0(x,y)}{\prod_m(1-F^m(v+y^m|x^m))}+\ds\sum_{r=0}^{n}\frac{f^r(y^r|x^r)}{(1-F^r(v+y^r|x^r))}P_{t,x,y}(\ell(t)=l)\\
		&=& \ds\sum_{r=0}^{n} f_{\tau^r}(0|x^r,y^r)P_{t,x,y}(\ell(t)=l)-f_{\tau^l}(0|x^l,y^l).
	\end{eqnarray*}
	Operating $D_{t,y}$ on $F_{\tau^l|l}(T-t|x,y)$
	\begin{eqnarray*}
		D_{t,y}F_{\tau^l|l}(T-t|x,y)&=&\frac{D_{t,y}\digamma_{T-t}(x,y)}{\digamma_\infty(x,y)}-\frac{\digamma_{T-t}(x,y)D_{t,y}\digamma_\infty(x,y)}{\digamma^2_\infty(x,y)}\\
		&=& -\frac{\digamma'_0(x,y)}{\digamma_\infty(x,y)}+\frac{\digamma_{T-t}(x,y)\digamma'_0(x,y)}{\digamma^2_\infty(x,y)}\\
		&=& f_{\tau^l|l}(0|x,y)(F_{\tau^l|l}(T-t|x,y)-1).
	\end{eqnarray*}
	This completes the proof of (iv).
	
	\noi (v) Follows from a direct calculation.
	
	\qed

	{\bf Acknowledgement:} The authors are grateful to Mrinal K. Ghosh and Anup Biswas for very useful discussions.
	
	\bibliographystyle{plain}

\end{document}